\documentclass{amsart}
\usepackage{amsthm,amsmath,amssymb}
\usepackage{graphicx}

\theoremstyle{plain}
\newtheorem{theorem}{Theorem}
\newtheorem{lemma}[theorem]{Lemma}
\newtheorem{corollary}[theorem]{Corollary}
\newtheorem{proposition}[theorem]{Proposition}

\theoremstyle{definition}
\newtheorem{definition}[theorem]{Definition}
\newtheorem{example}[theorem]{Example}

\theoremstyle{remark}
\newtheorem{remark}[theorem]{Remark}

\usepackage{enumerate}

\DeclareMathOperator{\val}{val}
\DeclareMathOperator{\rep}{rep}
\DeclareMathOperator{\relpos}{relpos}

\title{Behavior of digital sequences through exotic numeration systems}
\author{Julien Leroy
\qquad
Michel Rigo
\qquad
Manon Stipulanti}

\thanks{J. Leroy is an FNRS post-doctoral fellow.}
\thanks{M. Stipulanti is supported by a FRIA grant 1.E030.16.}
\address{Universit\'e de Li\`ege, Institut de math\'ematique, All\'ee de la d\'ecouverte 12 (B37),
4000 Li\`ege, Belgium\newline
J.Leroy@ulg.ac.be, M.Rigo@ulg.ac.be, M.Stipulanti@ulg.ac.be}

\begin{document}

\maketitle

\begin{abstract}
Many digital functions studied in the literature, e.g., the summatory function of the base-$k$ sum-of-digits function, have a behavior showing some periodic fluctuation. 
Such functions are usually studied using techniques from analytic number theory or linear algebra.
In this paper we develop a method based on exotic numeration systems and we apply it on two examples motivated by the study of generalized Pascal triangles and binomial coefficients of words.
\end{abstract}

\section{Introduction}

Many digital functions, e.g., the sum of the output labels of a finite transducer reading base-$k$ expansion of integers \cite{Heuberger}, have been extensively studied in the literature and exhibit an interesting behavior that usually involves some periodic fluctuation \cite{BR,Delange,Dumas13,Dumas14,GR,GT}. 
For instance, consider the archetypal sum-of-digits function $s_2$ for base-$2$ expansions of integers~\cite{trollope}. 
Its summatory function $$N\mapsto \sum_{j=0}^{N-1}s_2(j)$$ is counting the total number of ones occurring in the base-$2$ expansion of the first $N$ integers, i.e., the sum of the sums of digits of the first $N$ integers. 
Delange \cite{Delange} showed that there exists a continuous nowhere differentiable periodic function $\mathcal{G}$ of period $1$ such that 
\begin{equation}
    \label{eq:delange}
    \frac{1}{N} \sum_{j=0}^{N-1}s_2(j)=\frac{1}{2} \log_2 N+ \mathcal{G}(\log_2 N).
\end{equation}
For an account on this result, see, for instance, \cite[Thm.~3.5.4]{AS03}. 
The function $s_2$ has important structural properties. It readily satisfies $s_2(2n)=s_2(n)$ and $s_2(2n+1)=1+s_2(n)$ meaning that the sequence $(s_2(n))_{n\ge 0}$ is $2$-regular in the sense of Allouche and Shallit \cite{AS99}. A sequence $(u(n))_{n\ge 0}$ is {\em $k$-regular} if the $\mathbb{Z}$-module generated by its $k$-kernel, i.e., the set of subsequences $\{(u(k^jn+r))_{n\ge 0}\mid j\ge 0, r<k^j\}$, is finitely generated. This is equivalent to the fact that the sequence $(u(n))_{n\ge 0}$ admits a {\em linear representation}: there exist an integer $d\ge 1$, a row vector $\mathbf{r}$, a column vector $\mathbf{c}$ and matrices $\Gamma_0,\ldots,\Gamma_{k-1},$ of size $d$ such that, if the base-$k$ expansion of $n$ is $w_j\cdots w_0$, then $$u(n)=\mathbf{r}\, \Gamma_{w_0}\Gamma_{w_1}\cdots\Gamma_{w_j}\, \mathbf{c}.$$
For instance, the sequence $(s_2(n))_{n\ge 0}$ admits the linear representation
$$\mathbf{r}=
\begin{pmatrix}
    1&0\\
\end{pmatrix},\ 
\Gamma_0=\left(
\begin{array}{cc}
 1 & 0 \\
 0 & 1 \\
\end{array}
\right),\ 
\Gamma_1=\left(
\begin{array}{cc}
 1 & 1 \\
 0 & 1 \\
\end{array}
\right),\ 
\mathbf{c}=
\begin{pmatrix}
    0\\1\\
\end{pmatrix}.
$$
Many examples of $k$-regular sequences may be found in \cite{AS99,AS99II}.
Based on linear algebra techniques, Dumas \cite{Dumas13, Dumas14} provides general asymptotic estimates for summatory functions of $k$-regular sequences similar to~\eqref{eq:delange}. 
Similar results are also discussed by Drmota and Grabner \cite[Thm.~9.2.15]{BR} and by Allouche and Shallit \cite{AS03}. 

In this paper, we expose a new method to tackle the behavior of the summatory function $(A(N))_{N\ge 0}$ of a digital sequence $(s(n))_{n\ge 0}$. 
Roughly, the idea is to find two sequences $(r(n))_{n \geq 0}$ and $(t(n))_{n \geq 0}$, each satisfying a linear recurrence relation, such that $A(r(n)) = t(n)$ for all $n$. 
Then, from a recurrence relation satisfied by $(s(n))_{n\ge 0}$, we deduce a recurrence relation for $(A(n))_{n\ge 0}$ in which $(t(n))_{n\ge 0}$ is involved.
This allows us to find relevant representations of $A(n)$ in some exotic numeration system associated with the sequence $(t(n))_{n\ge 0}$. 
The adjective ``exotic'' means that we have a decomposition of particular integers as a linear combination of terms of the sequence $(t(n))_{n\ge 0}$ with possibly unbounded coefficients. 
We present this method on two examples inspired by the study of generalized Pascal triangle and binomial coefficients of words~\cite{LRS1,LRS2} and obtain behaviors similar to~\eqref{eq:delange}. 
The behavior of the first example comes with no surprise as the considered sequence is $2$-regular \cite{LRS2}. 
Nevertheless, the method provides an exact behavior although an error term usually appears with classical techniques. 
Furthermore, our approach also allows us to deal with sequences that do not present any $k$-regular structure, as illustrated by the second example. 

Let us make the examples a bit more precise (definitions and notation will be provided in due time). 
The \emph{binomial coefficient} $\binom{u}{v}$ of two finite words $u$ and $v$ in $\{0,1\}^*$ is defined as the number of times that $v$ occurs as a subsequence of $u$ (meaning as a ``scattered'' subword) \cite[Chap.~6]{Lot}. 
The sequence $(\mathsf{s}(n))_{n\ge 0}$ is defined from base-$2$ expansions by $s(0)=0$ and, for all $n\ge 1$,
\begin{equation}
    \label{eq:defS}
    \mathsf{s}(n):=\#\left\{v\in \rep_2(\mathbb{N}) \mid \binom{\rep_2(n-1)}{v}>0\right\}
\end{equation}
and the sequence $(\mathsf{s}_F(n))_{n\ge 0}$ is defined from Zeckendorf expansions by 
\begin{equation}
    \label{eq:defSF}
    \mathsf{s}_F(n):=\#\left\{v\in \rep_F(\mathbb{N}) \mid \binom{\rep_F(n)}{v}>0\right\},
\end{equation}
where $F$ stands for the Fibonacci numeration system.
We have the following results.

\begin{theorem}\label{thm:AsymptoBase2}
There exists a continuous and periodic function $\mathcal{H}$ of period 1 such that, for all $N\ge 1$, 
$$ \sum_{j=0}^{N} \mathsf{s}(j) =3^{\log_2 N}\ \mathcal{H}(\log_2 N ).$$ 
\end{theorem}

\begin{theorem}\label{thm:AsymptoFib}
Let $\beta$ be the dominant root of $X^3-2X^2-X+1$. 
There exists a continuous and periodic function $\mathcal{G}$ of period 1 such that, for $N \geq 3$,
$$
	\sum_{j=0}^N \mathsf{s}_F(j)  = c\, \beta^{\log_F N} \mathcal{G}(\log_F N) + o(\beta^{\lfloor \log_F N \rfloor}).
$$
\end{theorem}

In the last section of the paper, we present some conjectures in a more general context.  

\section{Summatory function of a $2$-regular sequence using particular $3$-decompositions}\label{sec:base2et3}

This section deals with the first example. 
For $n \in \mathbb{N}$, we let $\rep_2(n)$ denote the usual base-$2$ expansion of $n$. 
We set $\rep_2(0)=\varepsilon$ and get $\rep_2(\mathbb{N})=1\{0,1\}^*\cup \{\varepsilon\}$. 
We will consider the summatory function $$A(N):=\sum_{j=0}^N \mathsf{s}(j)$$ of the sequence $(\mathsf{s}(n))_{n\ge 0}$. 
The first few terms of $(A(N))_{N\ge 0}$ are
$$
0, 1, 3, 6, 9, 13, 18, 23, 27, 32, 39, 47, 54, 61, 69, 76, 81, 87, 96, 
107,117,\ldots
$$ 
The quantity $A(N)$ can be thought of as the total number of base-$2$ expansions occurring as subwords in the base-$2$ expansion of integers less than or equal to $N$ (the same subword is counted $k$ times if it occurs in the base-$2$ expansion of $k$ distinct integers). 
The sequence $(\mathsf{s}(n))_{n\ge 0}$ being 2-regular~\cite{LRS2}, asymptotic estimates of $(A(N))_{N\ge 0}$ could be deduced from~\cite{Dumas13}. 
However, as already mentioned, such estimates contain an error term. Applying our method, we get an exact formula for $A(N)$ given by Theorem~\ref{thm:AsymptoBase2}. 
To derive this result, we make an extensive use of a particular decomposition of $A(N)$ based on powers of $3$ that we call $3$-decomposition. 
These occurrences of powers of 3 come from the following lemma.

\begin{lemma}[{\cite[Lemma 9]{LRS1}}]
\label{lem:2^n 3^n}
For all $n \in \mathbb{N}$, we have $A(2^n)=3^n$.
\end{lemma}

\begin{remark}
The sequence $(\mathsf{s}(n))_{n\ge 0}$ also appears as the sequence {\tt A007306} of the denominators occurring in the Farey tree (left subtree of the full Stern--Brocot tree) which contains every (reduced) positive rational less than $1$ exactly once. 
Note that we can also relate $(\mathsf{s}(n))_{n\ge 0}$ to Simon's congruence where two finite words are equivalent if they share exactly the same set of subwords \cite{Simon}. 
\end{remark}

For the sake of presentation, we introduce the \emph{relative position} $\relpos_2(x)$ of a positive real number $x$ inside an interval of the form $[2^n,2^{n+1})$, i.e.,  
$$
\relpos_2(x):=\frac{x-2^{\lfloor \log_2 x\rfloor}}{2^{\lfloor \log_2 x\rfloor}}=2^{ \{ \log_2 x \} } - 1\in[0,1),
$$
where $\{x \} = x - \lfloor x \rfloor$ denotes the fractional part of any real number $x$. 
In Figure~\ref{fig:log}, the map $\lfloor \log_2 N \rfloor+\relpos_2(N)$ is compared with $\log_2 N$. 
Observe that both functions take the same value at powers of $2$ and the first one is affine between two consecutive powers of $2$. 
\begin{figure}[h!tb]
    \centering
    \scalebox{.5}{\includegraphics{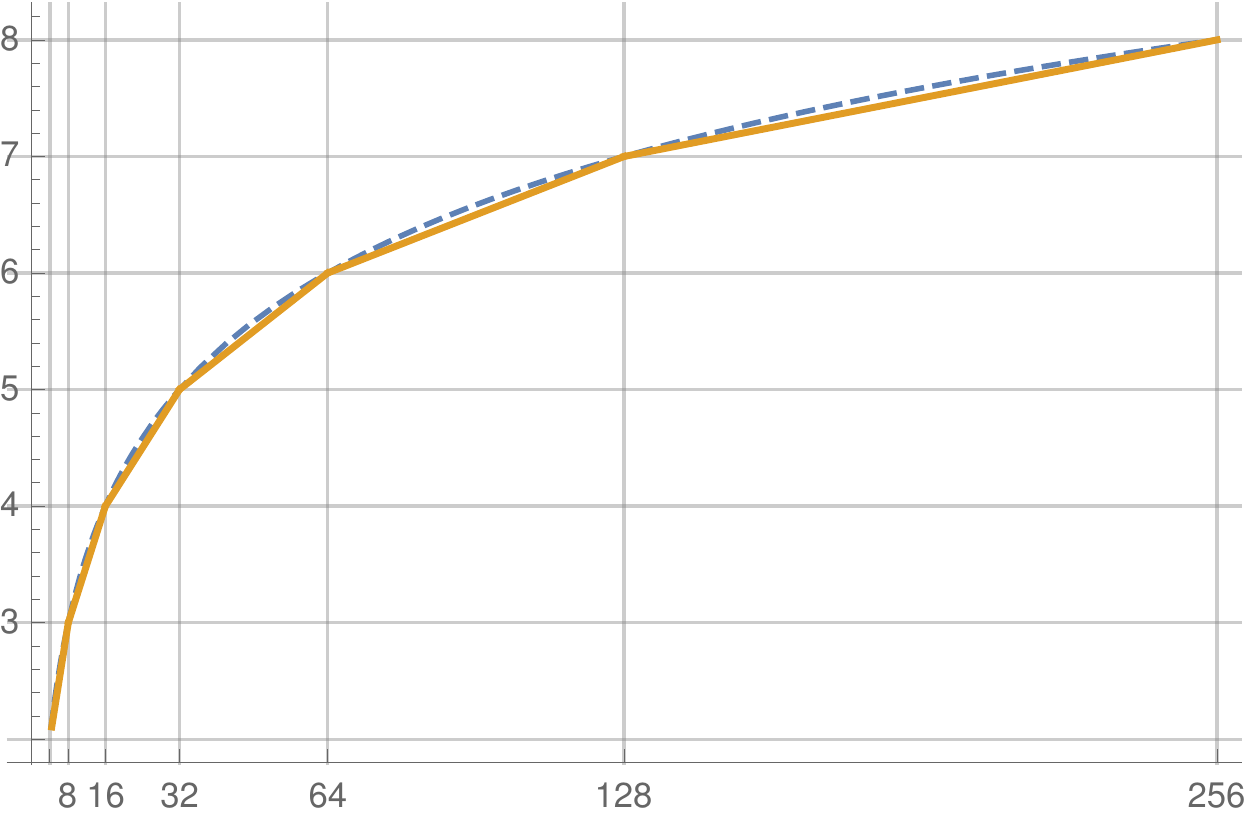}}
    \caption{The map $\lfloor \log_2 N \rfloor+\relpos_2(N)$.}
    \label{fig:log}
\end{figure}

In the rest of the section, we prove the following result which is an equivalent version of Theorem~\ref{thm:AsymptoBase2} when considering the function $\mathcal{H}$ defined by $\mathcal{H}(x)=\Phi(\relpos_2(2^x))$.

\begin{theorem}\label{thm:convergence}
There exists a continuous function  $\Phi$ over $[0,1)$ such that $\Phi(0)=1$, $\lim_{\alpha\to 1^-}\Phi(\alpha)=1$ and the sequence $(A(N))_{N\ge 0}$ satisfies, for all $N\ge 1$, 
$$
A(N)=3^{\log_2 N}\ \Phi(\relpos_2(N))=N^{\log_2 3}\ \Phi(\relpos_2(N)).
$$ 
\end{theorem}

The graph of $\Phi$ is depicted in Figure~\ref{fig:F} and we will show in Lemma~\ref{lem:rationalPhi} that $\Phi$ can be computed on a dense subset of $[0,1)$.
\begin{figure}[h!tb]
    \centering
    \scalebox{.49}{\includegraphics{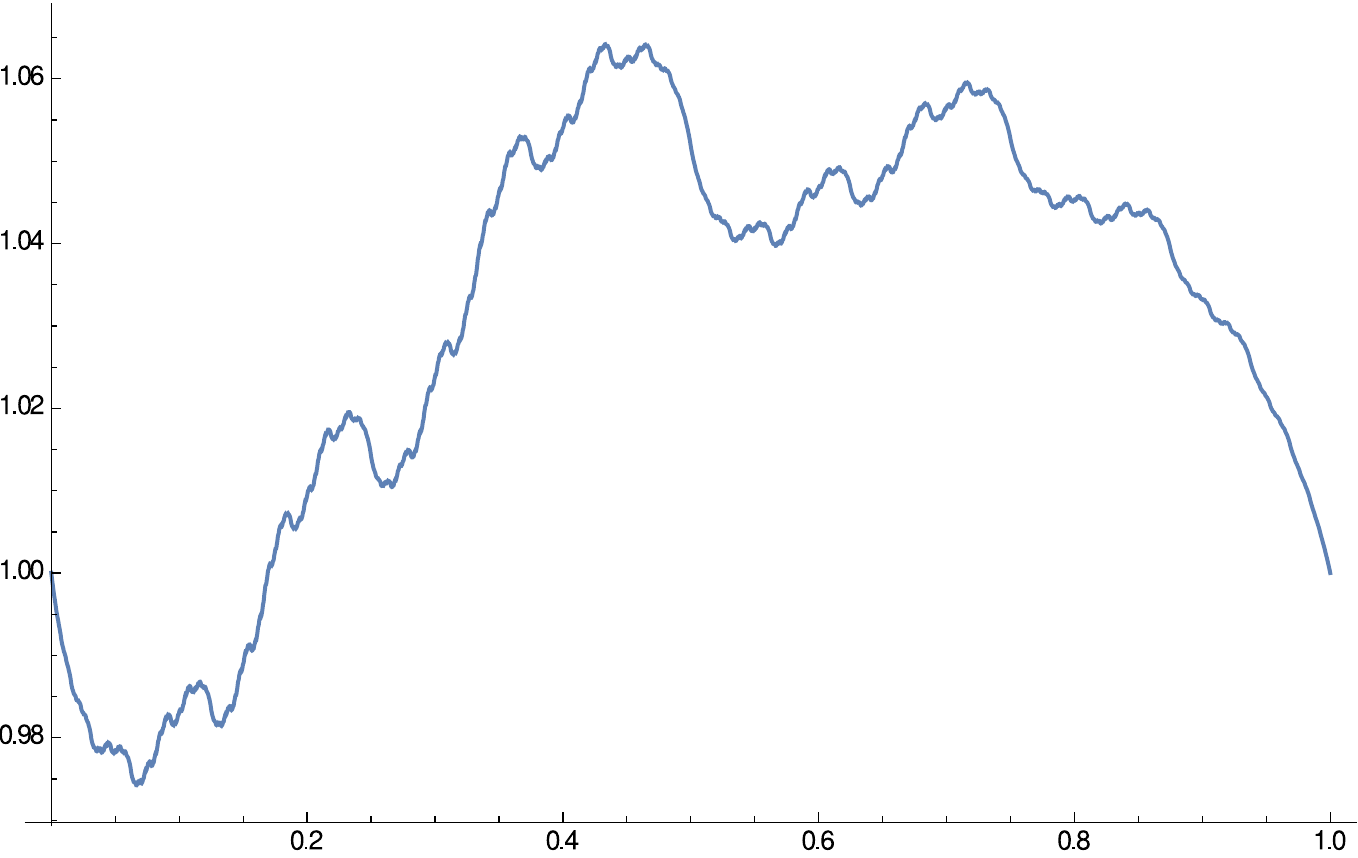}}
    \caption{The graph of $\Phi$.}
    \label{fig:F}
\end{figure}

Let us recall\footnote{With our current notation, the sequence was originally defined as $(\mathrm{s}(n+1))_{n\ge 0}$~\cite{LRS2}.
Considering a shifted version of it makes proofs simpler proofs.} the following result \cite{LRS2} which is our main tool. 

\begin{proposition}\label{pro:rec}
    The sequence $(\mathsf{s}(n))_{n\ge 0}$ satisfies $\mathsf{s}(0)=0$, $\mathsf{s}(1)=1$, $\mathsf{s}(2)=2$ and, for all $\ell\ge 1$ and $1\le r\le 2^\ell$, 
$$
\mathsf{s}(2^{\ell}+r)
=
\left\{
    \begin{array}{ll}
        \mathsf{s}(2^{\ell-1}+r)+\mathsf{s}(r),& \text{ if }1 \le r \le 2^{\ell-1};\\
        \mathsf{s}(2^{\ell+1}-r+1),& \text{ if } 2^{\ell-1} < r \le 2^\ell.\\
    \end{array}
\right.
$$
\end{proposition}

The following result directly follows by induction from Proposition~\ref{pro:rec}.

\begin{corollary}\label{cor:majs}
For all $n \geq 0$, we have $\mathsf{s}(n) \leq 2n$.
\end{corollary}

Proposition~\ref{pro:rec} also permits us to derive two convenient relations for $A(N)$ where powers of $3$ appear. 
This is the starting point of the $3$-decompositions mentioned above. 

\begin{lemma}\label{lem:recurrence}
Let $\ell\ge 1$. 
If $0\le r\le 2^{\ell-1}$, then
    $$A(2^\ell+r)=2\cdot 3^{\ell-1}+A(2^{\ell-1}+r)+A(r).$$
If $2^{\ell-1}< r< 2^{\ell}$, then
    $$A(2^\ell+r)=4\cdot  3^\ell-2\cdot 3^{\ell-1}-A(2^{\ell-1}+r')-A(r')\quad 
\text{ where }r'=2^\ell-r.$$
\end{lemma}

\begin{proof}
Let us start with the first case. 
If $r=0$, the result directly follows from Lemma~\ref{lem:2^n 3^n}. 
Now assume that $0< r\le 2^{\ell-1}$. 
Applying Proposition~\ref{pro:rec} and recalling that $\mathsf{s}(0)=0$, we get
\begin{eqnarray*}
    A(2^\ell+r)&=&A(2^\ell)+\sum_{j=1}^r \mathsf{s}(2^\ell+j)\\
&=&3^\ell+\sum_{j=1}^r \mathsf{s}(j) + \sum_{j=1}^r \mathsf{s}(2^{\ell-1}+j)\\
&=&3^\ell+A(r)+A(2^{\ell-1}+r)-A(2^{\ell-1})\\
&=&2\cdot 3^{\ell-1}+A(r)+A(2^{\ell-1}+r).
\end{eqnarray*}

Let us proceed to the second case with $2^{\ell-1}< r< 2^{\ell}$ and $r'=2^\ell-r$. 
Notice that $0< r'<2^{\ell-1}$. 
Applying Proposition~\ref{pro:rec}, we get
\begin{eqnarray*}
    A(2^\ell+r)=A(2^{\ell+1}-r')&=& A(2^{\ell+1})-\sum_{j=1}^{r'} \mathsf{s}(2^{\ell+1}-j+1)\\
&=&3^{\ell+1}-A(2^{\ell}+r')+A(2^{\ell})\\
&=&4\cdot  3^\ell-A(2^{\ell}+r').
\end{eqnarray*}
We may apply the first part of this lemma with $r'$ and thus get
$$
A(2^\ell+r)=4\cdot 3^\ell-2\cdot 3^{\ell-1}-A(r')-A(2^{\ell-1}+r').
$$
\end{proof}

\begin{corollary}\label{cor:32n}
    For all $n\ge 0$, $A(2n)=3 A(n)$.
\end{corollary}
\begin{proof}
Let us proceed by induction on $n\ge 0$. 
The result holds for $n\in\{0,1\}$. 
Thus consider $n \ge 2$ and suppose that the result holds for all $m < n$. 
Let us write $n=2^\ell+r$ with $\ell\ge 1$ and $0\le r < 2^\ell$. 
Let us first suppose that $0\le r \le 2^{\ell-1}$. 
Then, by Lemma~\ref{lem:recurrence}, we have
$$
3A(n)-A(2n) = 2\cdot 3^{\ell}+3A(2^{\ell-1}+r)+3A(r) - 2\cdot 3^{\ell}-A(2^{\ell}+2r)-A(2r).
$$
We conclude this case by using the induction hypothesis. 
Now suppose that $2^{\ell-1} < r < 2^\ell$. 
Then, by Lemma~\ref{lem:recurrence}, we have
\begin{align*}
3A(n)-A(2n) = &4\cdot 3^{\ell+1}-2\cdot 3^{\ell}-3A(2^{\ell-1}+r')-3A(r') \\
&- 4\cdot 3^{\ell+1}+2\cdot 3^{\ell}+A(2^{\ell}+2r')+A(2r')  
\end{align*}
where $r'=2^\ell-r$. 
We again conclude by using the induction hypothesis. 
\end{proof}

\subsection{$3$-decomposition of $A(n)$}\label{sec:3dec} 

Let us consider two examples to understand the forthcoming notion of $3$-decomposition. 
The idea is to iteratively apply Lemma~\ref{lem:recurrence} to derive a decomposition of $A(n)$ as a particular linear combination of powers of $3$. 
Indeed, each application of Lemma~\ref{lem:recurrence} provides a ``leading'' term of the form $2\cdot 3^{\ell-1}$ or $4\cdot 3^\ell-2\cdot 3^{\ell-1}$ plus terms where smaller powers of $3$ will occur. 
To be precise, the special case of $A(2^\ell+2^{\ell-1})$ gives, when applying the lemma twice, a term $4\cdot 3^{\ell-1}$ plus terms where smaller powers of $3$ will occur. 
We also choose to set $A(0)=0\cdot 3^0$ and $A(1)=1\cdot 3^0$.

\begin{example}
To compute $A(42)$, three applications of Lemma~\ref{lem:recurrence} yield
\begin{eqnarray*}
A(42)=A(2^5+10) &=& 2\cdot 3^{4}+A(2^{4}+10)+A(2^3+2); 	\\
A(2^{4}+10) 	&=& 4\cdot 3^4-2\cdot 3^3-A(2^3+6)-A(2^2+2);	\\
A(2^3+2)		&=& 2\cdot 3^{2}+A(2^2+2)+A(2).
\end{eqnarray*}
We thus get 
$$
A(42)=6\cdot 3^{4}-2\cdot 3^3-A(2^3+6)+2\cdot 3^{2}+3\cdot 3^0.
$$
At this stage, we already know that, in the forthcoming applications of the lemma, no other term in $3^4$ may occur because we are left with the decomposition of $A(2^3+6)$. 
Applying again Lemma~\ref{lem:recurrence} yields
\begin{eqnarray*}
A(2^3+6)	&=& 4\cdot 3^3-2\cdot 3^{2}-A(2^{2}+2)-3\cdot 3^0;	\\
A(2^2+2)	&=& 2\cdot 3+A(2+2)+A(2);	\\
A(4)=A(2^2)	&=& 2\cdot 3+A(2)+A(0) = 2\cdot 3+3\cdot 3^0.
\end{eqnarray*}
So we have $A(2^2+2)=4\cdot 3+6\cdot 3^0$, $A(2^3+6)=4\cdot 3^3-2\cdot 3^{2}-4\cdot 3-9\cdot 3^0$
and, finally,
\begin{equation}
    \label{eq:42}
    A(42)=6\cdot 3^{4}-6\cdot 3^3+4\cdot 3^{2}+4\cdot 3+12\cdot 3^0.
\end{equation}

Proceeding similarly with $A(84)$, we have 
\begin{eqnarray*}
A(84)=A(2^6+20) &=& 2\cdot 3^{5}+A(2^{5}+20)+A(2^4+4);		\\
A(2^{5}+20)		&=& 4\cdot 3^5-2\cdot 3^4-A(2^4+12)-A(2^3+4);		\\
A(2^4+4)		&=& 2\cdot 3^{3}+A(2^3+4)+A(4) 				\\
				&=& 2\cdot 3^{3}+A(2^3+4)+ 2\cdot 3+3\cdot 3^0.
\end{eqnarray*}
We thus get 
$$
A(84)=6\cdot 3^{5}-2\cdot 3^4-A(2^4+12)+2\cdot 3^{3}+2\cdot 3+3\cdot 3^0
$$
and
\begin{eqnarray*}
A(2^4+12)	&=& 4\cdot 3^4-2\cdot 3^{3}-A(2^{3}+4)-A(4);	\\
A(2^3+4)	&=& 2\cdot 3^2+A(2^2+4)+A(4);			\\
A(2^2+4)	&=& A(2^3)=2\cdot 3^2+A(4)+A(0);			\\
A(4)		&=& 2\cdot 3+3\cdot 3^0.
\end{eqnarray*}
Thus, we get $A(2^3+4)=4\cdot 3^2+4\cdot 3+6\cdot 3^0$ and finally
\begin{equation}
    \label{eq:84}
    A(84)=6\cdot 3^{5}-6\cdot 3^4+4\cdot 3^{3}+4\cdot 3^2+8\cdot 3+12\cdot 3^0.
\end{equation}

If we compare \eqref{eq:42} and \eqref{eq:84}, we may already notice that the same leading coefficients $6,-6$ and $4$ occur in front of the dominating powers of $3$.
\end{example}

\begin{definition}[$3$-decomposition]\label{def:3dec1}
Let $n\ge 2$. 
Iteratively applying Lemma~\ref{lem:recurrence} provides a unique decomposition of the form 
$$
A(n)=\sum_{i=0}^{\ell_2(n)} a_i(n)\, 3^{\ell_2(n)-i},
$$
where $a_i(n)$ are integers, $a_{0}(n)\neq 0$ and $\ell_2(n)$ stands for $\lfloor \log_2 n \rfloor$ or $\lfloor \log_2 n \rfloor-1$ (depending on the fact that $n=2^{\lfloor \log_2 n \rfloor}+r$ with $2^{\lfloor \log_2 n \rfloor-1}<r<2^{\lfloor \log_2 n \rfloor}$ or $0\le r\le 2^{\lfloor \log_2 n \rfloor-1}$ respectively). 
We say that the word 
$$
\mathsf{3dec}(A(n)):=a_0(n) \cdots a_{\ell_2(n)}(n)
$$ 
is the {\em $3$-decomposition} of $A(n)$. 
Observe that when the integer $n$ is clear from the context, we simply write $a_i$ instead of $a_i(n)$. For the sake of clarity, we will also write $(a_0(n), \ldots, a_{\ell_2(n)}(n))$.
\end{definition}

As an example, we have $\ell_2(84)=5$ and, using~\eqref{eq:84}, the $3$-decomposition of $A(84)$ is $(6, -6, 4, 4, 8, 12)$. 
See also Table~\ref{tab:Adecomp}. 
Also notice that the notion of $3$-decomposition is only valid when the values taken by the sequence $(A(N))_{N\ge 0}$ are concerned. 
For instance, the $3$-decomposition of $5$ is not defined because $5 \notin \{A(n) \mid n \in \mathbb{N}\}$. 
\begin{table}[h!tb]
$$\begin{array}{c|rrrr|r}
n & a_0(n) & a_1(n) & a_2(n) & a_3(n)& A(n) \\
\hline
2& 3 &  &  &  &3\times 1=3\\
3& 6 &  &  &  &6\times 1=6\\
4& 2 & 3 &  & &2\times 3+3\times 1=9 \\
5& 2 & 7 &  & &2\times 3+7\times 1=13 \\
6& 4 & 6 &  &  &4\times 3+6\times 1=18\\
7& 4 & -2 & -7 & & 4\times 3^2-2\times 3-7\times 1=23 \\
\vdots & & & & & \\
20& 2 & 4 & 6 & 9 & 2\times 3^3+4\times 3^2+6\times 3+9\times 1=117\\
\end{array}$$
\caption{The $3$-decomposition of $A(2),A(3),\ldots$}
    \label{tab:Adecomp}
\end{table}

\begin{remark}\label{rem:cases}
Assume that we want to develop $A(n)$ using only Lemma~\ref{lem:recurrence}, i.e., to get the $3$-decomposition of $A(n)$. Several cases may occur.
\begin{enumerate}[(i)]
\item If $\rep_2(n)=10u$, with $u\in\{0,1\}^*$ possibly starting with $0$, then we apply the
  first part of Lemma~\ref{lem:recurrence} and we are left with evaluations of $A$ at integers whose base-$2$ expansions are shorter and given by 
$1u \text{ and } \rep_2(\val_2(u))$. 
Note that $\rep_2(\val_2(u))$ removes the possible leading zeroes in front of $u$.
\item If $\rep_2(n)=11u$, with $u\in\{0,1\}^*\setminus 0^*$, i.e., $u$ contains at least one $1$, then we apply the
  second part of Lemma~\ref{lem:recurrence} and we are left with evaluations of $A$ at integers whose base-$2$ expansions are shorter and given by 
$1u' \text{ and } \rep_2(\val_2(u'))$ 
where $u'\in\{0,1\}^*$ has the same length as $u$ and satisfies $\val_2(u')=\val_2(h(u))+1$ where $h$ is the involutory morphism exchanging $0$ and $1$. 
As an example, if $u=01011000$, then $h(u)=10100111$ and $u'=10101000$. 
If we mark the last occurrence of $1$ in $u$ (such an occurrence always exists): $u=v10^n$ for some $n\ge 0$, then $u'=h(v)10^n$.
\item If $\rep_2(n)=110^k$, then we will apply the first part of Lemma~\ref{lem:recurrence} and we are left with evaluations of $A$ at integers whose base-$2$ expansions are given by $10^{k+1}$ and $10^k$. 
This situation seems not so nice : we are left with a word $10^{k+1}$ of the same length as the original one $110^k$. 
However, the next application of Lemma~\ref{lem:recurrence} provides the word $10^k$ and the computation easily ends with a total number of calls to this lemma equal to $k+3$, namely the computations of $A(2^{k+1}), A(2^{k}),\ldots, A(2^0), A(0)$ are needed. 
This situation is not so bad since the numbers of calls to Lemma~\ref{lem:recurrence} to evaluate $A$ at integers with base-$2$ expansions of the same length can be equal. 
For instance, the computation of $A(12)$ requires the computations of $A(8), A(4), A(2), A(1), A(0)$ and the one of $A(14)$ requires the computations of $A(6), A(4), A(2), A(1), A(0)$.  
\end{enumerate}
\end{remark}

As already observed with Equations~\eqref{eq:42} and~\eqref{eq:84}, the 3-decompositions of $42$ and $84$ share the same first digits.
The next lemma states that this is a general fact. 
Roughly speaking, if two integers $m,n$ have a long common prefix in their base-$2$ expansions, then the most significant coefficients in the corresponding $3$-decompositions of $A(m)$ and $A(n)$ are the same.
\begin{lemma}\label{lem:conv}
Let $u\in\{0,1\}^*$ be a finite word of length at least 2. 
For all finite words $v,v'\in \{0,1\}^*\setminus 0^*$, the $3$-decompositions of $A(\val_2(1uv))$ and $A(\val_2(1uv'))$ share the same coefficients $a_0,\ldots,a_{|u|-2}$.
\end{lemma}

\begin{proof}
It is a direct consequence of Lemma~\ref{lem:recurrence}. 
Proceed by induction on the length of the words. 
The word $u$ is of the form $0^{n_1}10^{n_2}1\cdots 10^{n_k}$ with $k\ge 1$ and $n_1,\ldots,n_k\ge 0$. 
If $n_1>0$, due to Lemma~\ref{lem:recurrence}, $A(\val_2(1uv))$ is decomposed as
$$2\cdot 3^{|u|+|v|-1}+A(\val_2(10^{n_1-1}10^{n_2}1\cdots 10^{n_k}v))+A(\val_2(10^{n_2}1\cdots 10^{n_k}v)).$$
Proceeding similarly, $A(\val_2(1uv'))$ is decomposed as
$$2\cdot3^{|u|+|v'|-1}+A(\val_2(10^{n_1-1}10^{n_2}1\cdots 10^{n_k}v'))+A(\val_2(10^{n_2}1\cdots 10^{n_k}v')).$$
The first term in these two expressions will equally contribute to the coefficient $a_0$ in the two $3$-decompositions. 
For the last two terms, we may apply the induction hypothesis. 
If $n_1=0$, applying again Lemma~\ref{lem:recurrence} to $A(\val_2(1uv))$ gives 
\begin{align*}
4\cdot3^{|u|+|v|}&-2\cdot3^{|u|+|v|-1}+A(\val_2(11^{n_2}0\cdots 01^{n_k}h(x)10^t))\\
&+A(\val_2(1^{n_i}01^{n_{i+1}}0\cdots 01^{n_k}h(x)10^t))
\end{align*}
where $v=x10^t$ and $i$ is the smallest index such that $n_i>0$.
We can conclude in the same way as in the case $n_1>0$.
\end{proof}

\begin{example}
 Take $\rep_2(745)=  1(01110)1001$ and $\rep_2(5904)=1(01110)0010000$. 
 If we compare the $3$-decompositions of $A(745)$ and $A(5904)$, they share the same first four coefficients.
$$\begin{array}{r||cccc|cccccccc}
n&a_0 & a_1 & a_2 & a_3 & a_4& a_5& a_6 & a_7 & a_8 & a_9 & a_{10} & a_{11}  \\
\hline
745&6& 2& -4& -12& 12& -42& -10& 32& 121 \\
5904&6& 2& -4& -12& -16& 14& 14& 28& 60& 60& 60& 90 \\
\end{array}$$
\end{example}

\begin{example}
    In this second example, we show that the assumption that $v\not\in 0^*$ is important. 
    Consider $\rep_2(448)=111000000$ and $\rep_2(449)=111000001$. 
    Even though these two words have the same prefix of length $8$, the third coefficient of the $3$-decompositions of $A(448)$ and $A(449)$ differ.
$$\begin{array}{r||cc|cccccccccc}
n&a_0 & a_1 & a_2 & a_3 & a_4& a_5& a_6 & a_7 & a_8   \\
\hline
448 & 4 & -2 & -4 & -6 & -6 & -6 & -6 & -6 & -9 \\
449 & 4 & -2 & -6 & -6 & 6 & 6 & 6 & 6 & 31 \\
\end{array}$$
\end{example}

The idea in the next three definitions is that $\alpha$ gives the relative position of an integer in the interval $[2^{n+1},2^{n+2})$.

\begin{definition}
    Let $\alpha$ be a real number in $[0,1)$. 
    Define the sequence of finite words $(w_n(\alpha))_{n\ge 1}$ where
$$w_n(\alpha):=(\rep_2\left(2^n+\lfloor \alpha 2^n \rfloor\right))1.$$
Roughly, $w_n(\alpha)$ is a word of length $n+2$ and its relative position amongst the words of length $n+2$ in $1\{0,1\}^*$ is given by an approximation of $\alpha$. 
\end{definition}
We add an extra $1$ as least significant digit for convenience (i.e., to avoid the third case of Remark~\ref{rem:cases}). 
The sequence $(w_n(\alpha))_{n\ge 1}$ converges to the infinite word $1\rep_2(\alpha)$ where $\rep_2(\alpha)$ is the infinite word $d_1d_2d_3\cdots$ over $\{0,1\}$ with the $d_i$'s not all eventually equal to $1$ and $\sum_{i\ge 1} d_i 2^{-i}=\alpha$. 
In particular, we may apply Lemma~\ref{lem:conv} to $w_n(\alpha)=1d_1\cdots d_{n}1$ and $w_{n+1}(\alpha)=1d_1\cdots d_{n}d_{n+1}1$ with $u=d_1\cdots d_{n}$, $v=1$ and $v'=d_{n+1}1$.

\begin{definition}\label{def:22}
    Let $\alpha$ be a real number in $[0,1)$. 
    Define the sequence $(e_n(\alpha)))_{n\ge 1}$ where $$e_n(\alpha):=\val_2(w_n(\alpha))=2^{n+1}+2\lfloor \alpha 2^n \rfloor+1.$$
\end{definition}
Note that $e_n(\alpha)$ only takes odd integer values in $[2^{n+1}+1,2^{n+2}-1]$ and 
\begin{equation}
    \label{eq:relposalpha}
    \relpos_2 (e_n(\alpha))=\frac{\lfloor \alpha 2^n \rfloor}{2^n}+\frac{1}{2^{n+1}}\to \alpha
\end{equation}
as $n$ tends to infinity.

\begin{definition}\label{def:3dec2}
Let $\alpha$ be a real number in $[0,1)$. 
We consider the sequence of finite words $(\mathsf{3dec}(A(e_n(\alpha))))_{n\ge 1}$. 
Thanks to Lemma~\ref{lem:conv}, this sequence of finite words converges to an infinite sequence of integers denoted by
$$\mathbf{a}(\alpha)=a_0(\alpha)\, a_1(\alpha)\, \cdots.$$
\end{definition}

\begin{example} Take $\alpha=\pi-3$. 
The sequence $(w_n(\alpha))_{n\ge 1}$ converges to $$10010010000111111011\cdots.$$
$$
    \begin{array}{c|c||ccccccccccc}
n & e_n(\alpha) & a_0 & a_1 & a_2 & a_3 &\cdots \\
\hline
 1 & 5 & 2 & 7 &  &  &  &  &  &  &  &  &
    \\
 2 & 9 & 2 & 2 & 8 &  &  &  &  &  &  &  &
    \\
 3 & 19 & 2 & 6 & -2 & 5 &  &  &  &  &  &  & 
   \\
 4 & 37 & 2 & 6 & -6 & 6 & 15 &  &  &  &  &  &  \\
 5 & 73 & 2 & 6 & -6 & 2 & 8 & 31 &  &  &  &  &  \\
 6 & 147 & 2 & 6 & -6 & 2 & 24 & -8 & 14 &  &  &  &  \\
 7 & 293 & 2 & 6 & -6 & 2 & 24 & -24 & 22 & 53 &  &  &  \\
 8 & 585 & 2 & 6 & -6 & 2 & 24 & -24 & 6 & 30 & 116 &  &  \\
 9 & 1169 & 2 & 6 & -6 & 2 & 24 & -24 & 6 & 30 & 30 & 131 &  \\
 10 & 2337 & 2 & 6 & -6 & 2 & 24 & -24 & 6 & 30 & 30 & 30 & 146 \\
\end{array}$$
Hence the first terms of the sequence $\mathbf{a}(\alpha)$ are $2,6,-6,2,24,-24,6,30$. 
At each step, all coefficients are fixed except for the last two ones (see Lemma~\ref{lem:conv}). 
\end{example}

\subsection{Definition of the function $\Phi$}

We will first introduce an auxiliary function $\Phi(\alpha)$, for $\alpha\in[0,1)$, defined as the limit of a converging sequence of step functions built on the $3$-decomposition of $A(e_n(\alpha))$. 
For all $n\ge 1$, let $\phi_n$ be the function defined by 
$$
\phi_n(\alpha) = \frac{A(e_n(\alpha))}{3^{\log_2(e_n(\alpha))}} \quad \text{for } \alpha\in[0,1). 
$$

\begin{proposition}\label{pro:limexi}
The sequence $(\phi_n)_{n\ge 1}$ uniformly converges to the function $\Phi$ defined, for $\alpha\in[0,1)$, by
$$
\Phi(\alpha) = \begin{cases} \dfrac{1}{3^{1+\log_2(\alpha+1)}} \sum\limits_{i=0}^{+\infty} \dfrac{a_i(\alpha)}{3^{i}}, & \text{if } \alpha < 1/2; \\
\dfrac{1}{3^{\log_2(\alpha+1)}} \sum\limits_{i=0}^{+\infty} \dfrac{a_i(\alpha)}{3^{i}}, & \text{if } \alpha \ge 1/2.
\end{cases}
$$
\end{proposition}

\begin{remark}
If the reader wonders about the difference between the exponents in the definition of $\Phi$, observe that, if $\alpha$ tends to $(1/2)^+$, then the $3$-decomposition of $\alpha$ converges to $(z_n)_{n\ge 0}= 4, -6, -2, 4, 4, 4, \ldots$ and $\sum_{i=0}^{+\infty} (z_i/3^i)=2$. 
If $\alpha$ tends to $(1/2)^-$, then the $3$-decomposition of $\alpha$ converges to $(z'_n)_{n\ge 0}= 6, 2, -4, -4, -4, \ldots$ and $\sum_{i=0}^{+\infty} (z'_i/3^i)=6$. 
The continuity of $\Phi$ will be discussed in the proof of Theorem~\ref{thm:convergence}.
\end{remark}

To visualize the uniform convergence stated in Proposition~\ref{pro:limexi}, we have depicted the first functions $\phi_2,\ldots,\phi_9$ in Figure~\ref{fig:phin}. 
For instance, $e_2(\alpha)\in\{9,11,13,15\}$ explaining the four subintervals defining the step function $\phi_2$.
\begin{figure}[h!tb]
    \centering
    \scalebox{.29}{\includegraphics{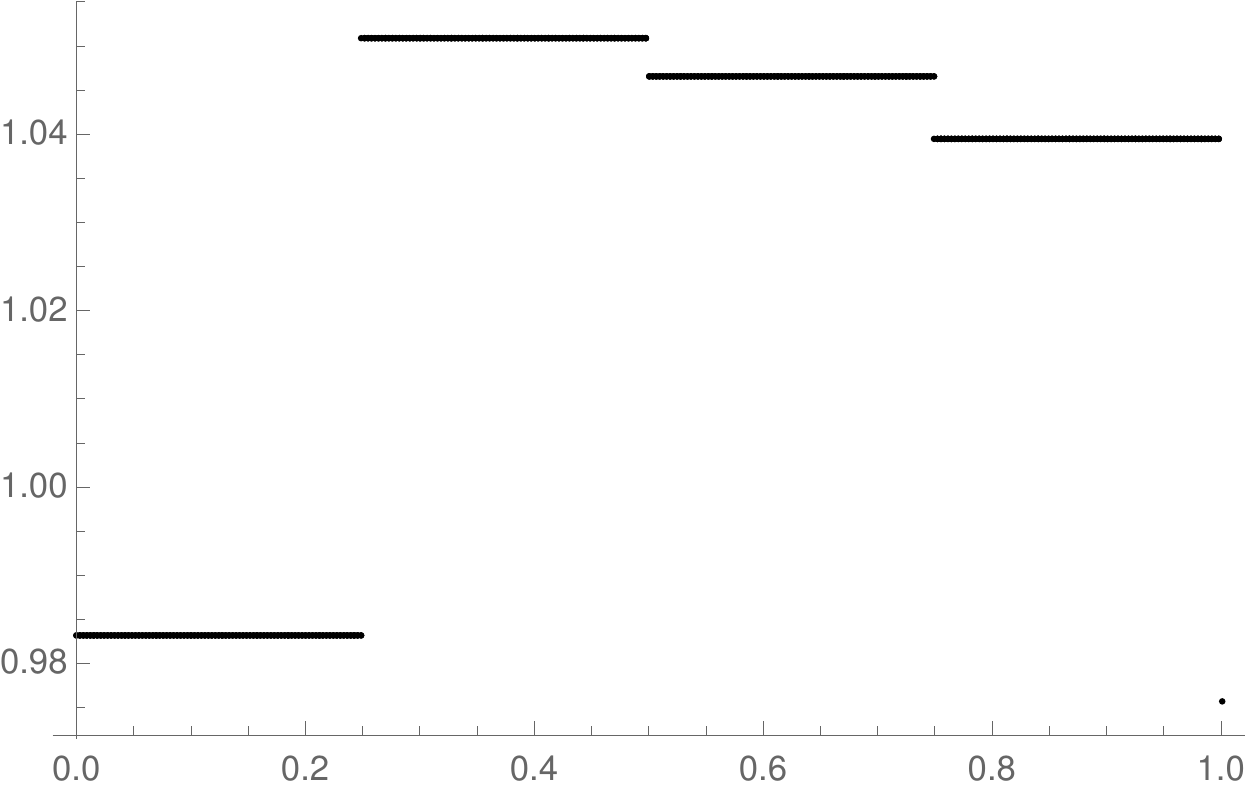}}\quad \scalebox{.29}{\includegraphics{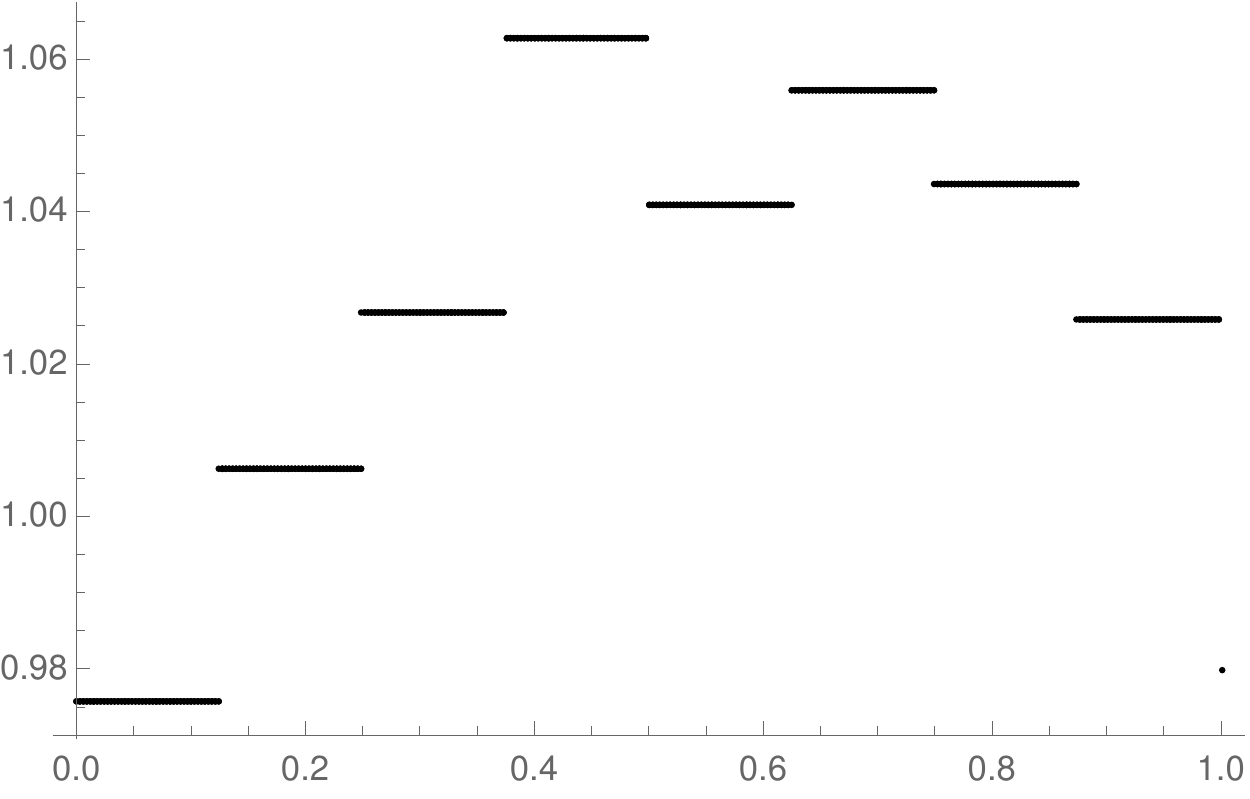}}\quad \scalebox{.29}{\includegraphics{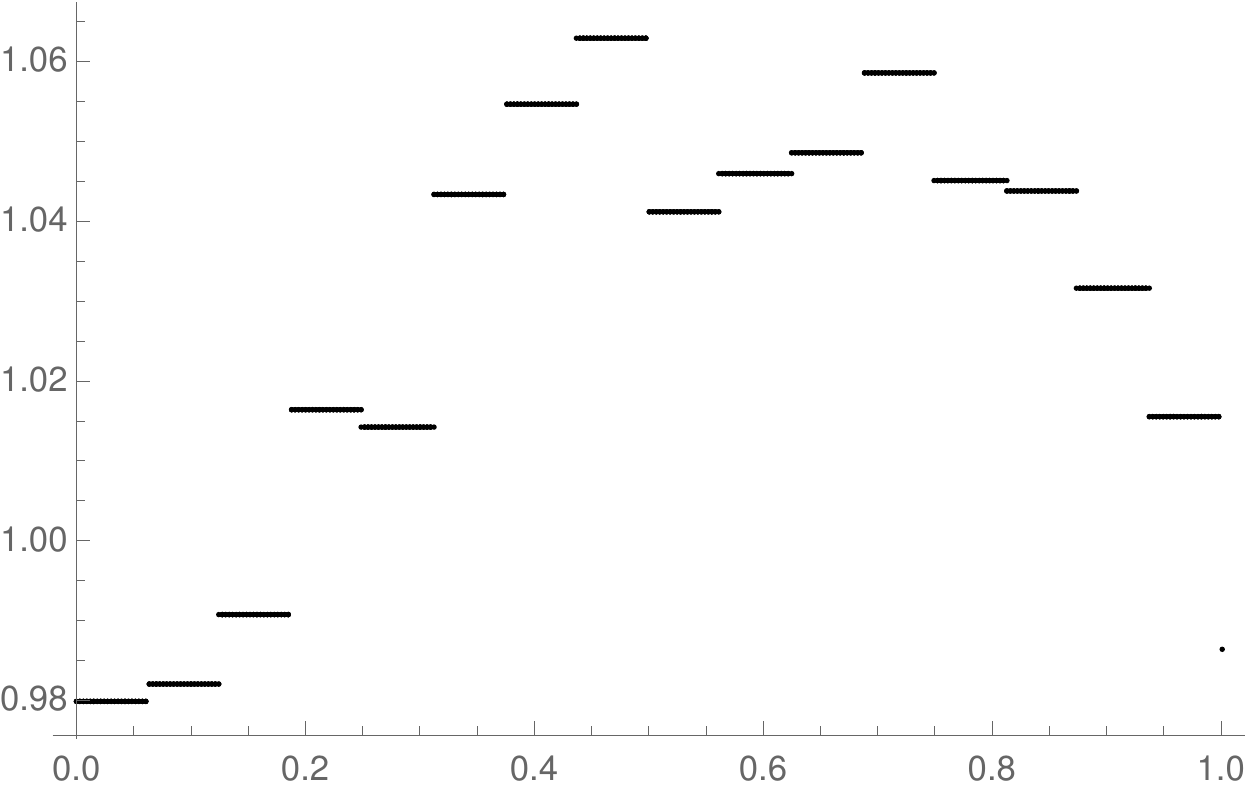}}\\

\scalebox{.29}{\includegraphics{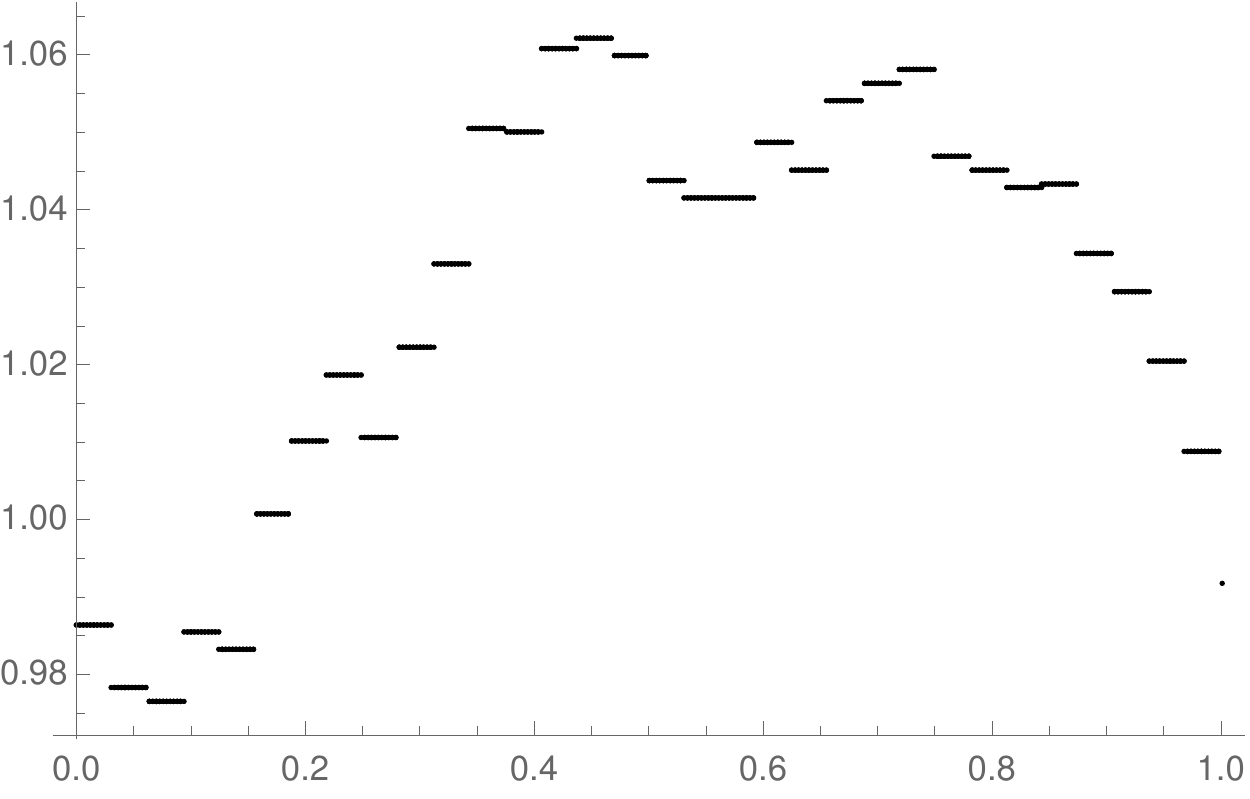}}\quad \scalebox{.29}{\includegraphics{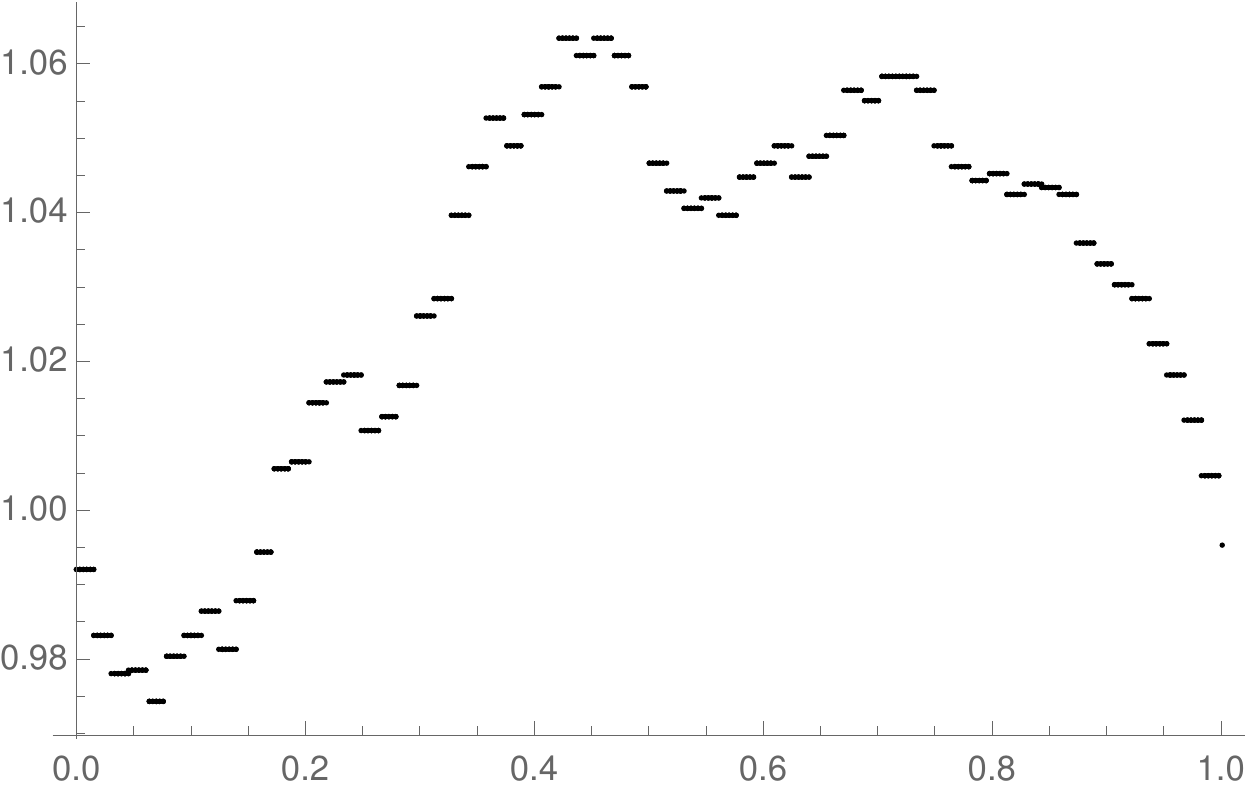}}\quad \scalebox{.29}{\includegraphics{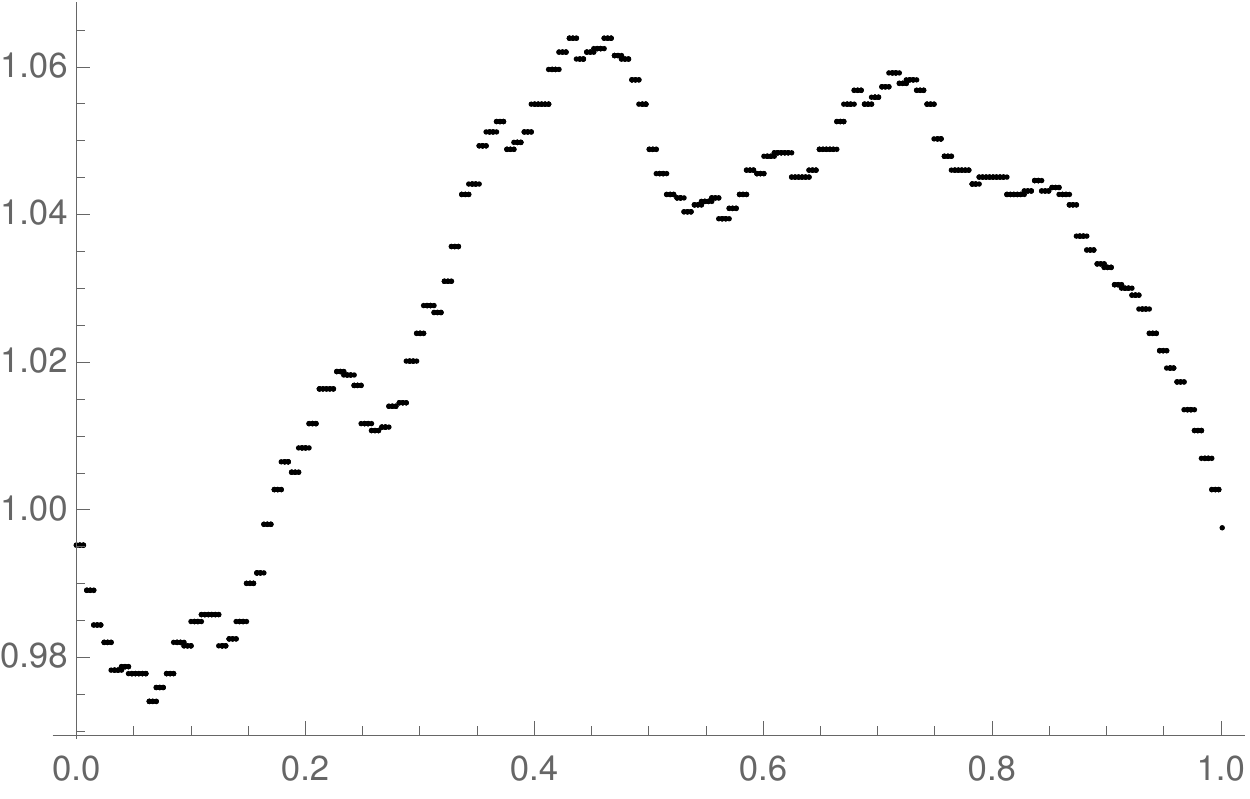}}\\

\scalebox{.29}{\includegraphics{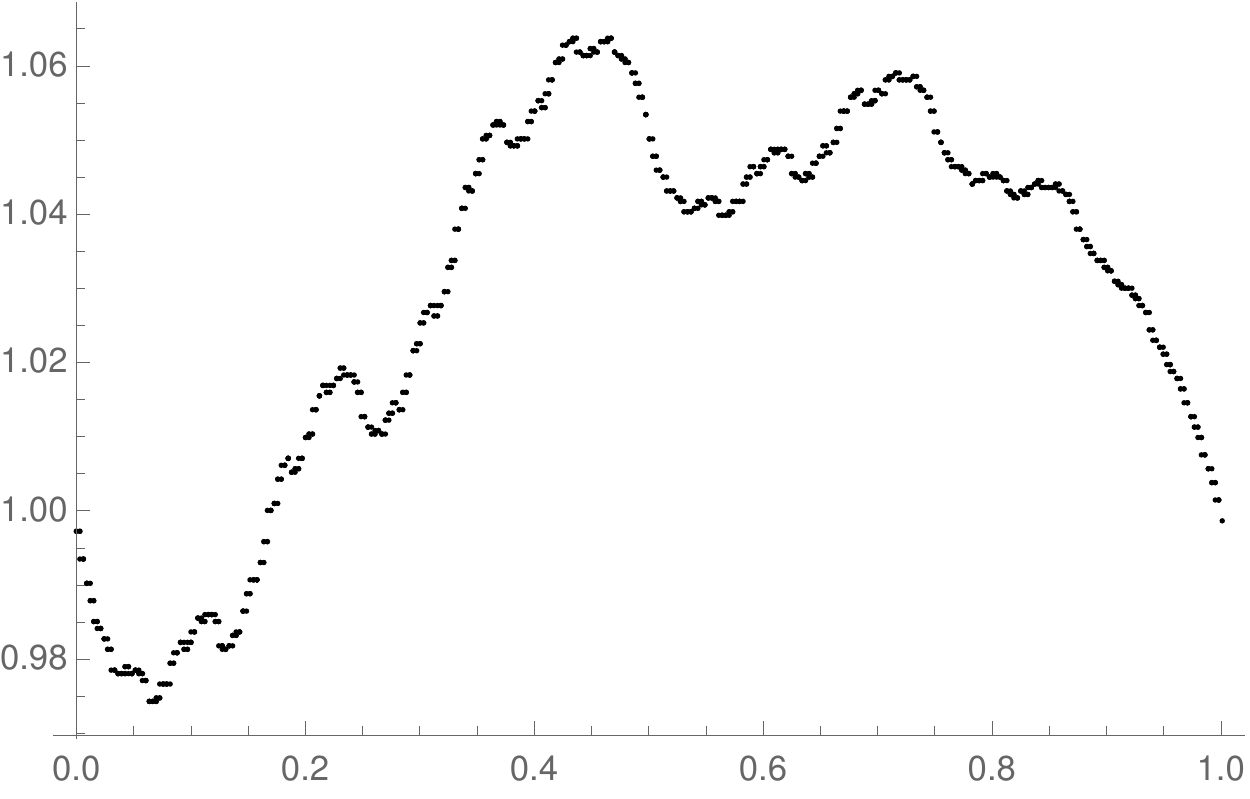}}\quad \scalebox{.29}{\includegraphics{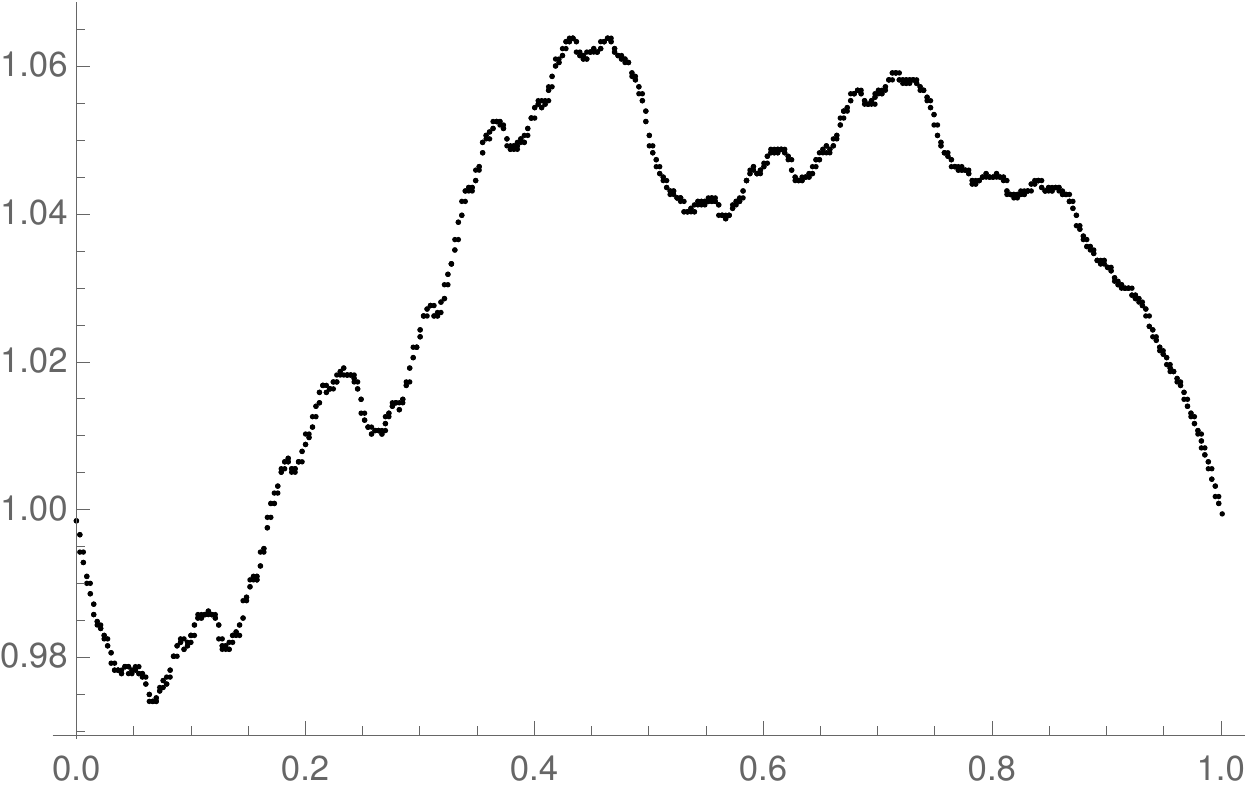}}
    \caption{Representation of $\phi_2,\ldots,\phi_9$.}
    \label{fig:phin}
\end{figure}

To ensure convergence of the series that we will encounter, we need some very rough estimate on the coefficients occurring in $\mathsf{3dec}(A(e_n(\alpha)))$.

\begin{lemma}\label{lem:estimate}
Using notation of Definitions \ref{def:3dec1} and \ref{def:3dec2}, for all $n\ge 2$ and for $0\le i\le \ell_2(n)$, we have 
$$
|a_i(n)|\le 10 \cdot 2^i.
$$
In particular, for all $\alpha \in [0,1)$ and all $i\ge 0$, we have
$$
|a_i (\alpha)|\le 10\cdot 2^i.
$$
\end{lemma}

\begin{proof}
Let us write $n = 2^\ell + r$ with $\ell \geq 1$ and $0 \leq r < 2^\ell$. 
Using Definition~\ref{def:3dec1}, let us write
$$
A(n)=\sum_{j=0}^{\ell_2(n)} a_j(n)\, 3^{\ell_2(n)-j}
$$
where $a_j(n)$ are integers, $a_{0}(n)\neq 0$. 
Observe that we have $\ell_2(n) \in \{\ell,\ell-1\}$.
Let us fix some $i \in \{0,1,\dots,\ell_2(n)\}$.
By Lemma~\ref{lem:recurrence}, terms of the form 
\begin{equation}\label{eq:form}
\begin{array}{ll}
A(2^{\ell_2(n)-i+1}+r') & \text{where } r'\in\{0,\ldots,2^{\ell_2(n)-i}\}, \text{ or }\\
A(2^{\ell_2(n)-i}+r'') & \text{where } r''\in\{2^{\ell_2(n)-i-1}+1,\ldots, 2^{\ell_2(n)-i}-1\}
\end{array}
\end{equation}
are the only ones possibly contributing to $a_i(n)$. Those of the first (resp., second) form yield $2\cdot 3^{\ell_2(n)-i}$ (resp., $4\cdot 3^{\ell_2(n)-i}$). 
Observe that for a term $A(2^{\ell_2(n)-i+1}+r')$ of the first form with $2^{\ell_2(n)-i-1} < r' \leq 2^{\ell_2(n)-i}$, a second application of Lemma~\ref{lem:recurrence} gives, in addition to $2\cdot 3^{\ell_2(n)-i}$, the term $A(2^{\ell_2(n)-i}+r')$, which is of the second form.
Together, these terms give $6 \cdot 3^{\ell_2(n)-i}$.
 
Our aim is now to understand, starting from $A(2^\ell +r)$, how the successive applications of Lemma~\ref{lem:recurrence} lead to terms of the form~\eqref{eq:form}.
Observe that the successive applications of the lemma can give terms of the form $A(2^p + r')$ where $r'$ can take several values for a given value of $p$. 
This is the reason why we consider a second index $q$ in the sum below.

Let us describe a transformation process starting from a linear combination of the form 
\[
	\sum_{\substack{0 \leq p \leq k \\ 0 \leq q \leq s_p}} 
	x_{p,q} A(2^p+r_{p,q}),
\]
where $k > \ell_2(n)-i+1$ and, for all $p$ and $q$, $s_p \in \mathbb{N}$, $x_{p,q} \in \mathbb{Z}$ and $r_{p,q} \in \{0,1,\dots,2^p-1\}$.
Applying Lemma~\ref{lem:recurrence} to every term of the form $A(2^p+r_{p,q})$ with $p< \ell_2(n)-i$ will provide terms of the form $A(2^{p'}+r')$ with $p' \leq p$ and $r' < 2^{p'}$.
Hence these terms are not of the form~\eqref{eq:form} and thus will never contribute to $a_i(n)$.
Applying the lemma to every term of the form $A(2^p+r_{p,q})$ with $p> \ell_2(n)-i+1$ gives a linear combination of $3^p$ and $3^{p-1}$ together with a linear combination of the form $x_1 A(2^{p_1}+r_{p_1}) + x_2 A(2^{p_2}+r_{p_2})$ with $p_1 < p_2 \leq p$ and $x_1,x_2 \in \{-1,1\}$.
Observe that $p_2=p$ if and only if $r_{p,q} = 2^{p-1}$.
In this case we get $p_1 = p-1$, $r_1 = r_2 = 0$ and the terms $A(2^p) = 3^p$ and $A(2^{p-1}) = 3^{p-1}$.
Therefore, applying Lemma~\ref{lem:recurrence} to all terms of the form $A(2^p+r_{p,q})$ with $p>\ell_2(n)-i+1$ gives a linear combination of the form
\[
	\sum_{j= \ell_2(n)-i+1}^{k} y_j 3^j 
	+
	\sum_{\substack{0 \leq p < k \\ 0 \leq q \leq t_p}} 
	y_{p,q} A(2^p+r'_{p,q}),
\]
where for all $j$, $y_{j} \in \mathbb{Z}$ and for all $p$ and $q$, $t_p \in \mathbb{N}$, $y_{p,q} \in \mathbb{Z}$ and $r'_{p,q} \in \{0,1,\dots,2^p-1\}$ and where 
\[
	\sum_{\substack{0 \leq p < k \\ 0 \leq q \leq t_p}} |y_{p,q}|
	\leq 
	2 \sum_{\substack{0 \leq p \leq k \\ 0 \leq q \leq s_p}} |x_{p,q}|.
\] 
So we get some information about how behave the coefficients when applying once the transformation process. 
Starting from the particular combination $1\cdot A(2^\ell+r)$ and iterating this process $\ell - \ell_2(n)+i-1$ times, we thus obtain a linear combination of the form 
\[
	\sum_{j= \ell_2(n)-i+1}^{\ell_2(n)} y_j 3^j 
	+
	\sum_{\substack{0 \leq p \leq \ell_2(n)-i+1 \\ 0 \leq q \leq t_p}} 
	y_{p,q} A(2^p+r'_{p,q}),
\]
where 
\[
	\sum_{\substack{0 \leq p \leq \ell_2(n)-i+1 \\ 0 \leq q \leq t_p}} |y_{p,q}|
	\leq 2^{\ell - \ell_2(n) +i-1}  \cdot 1 
	\leq 2^{i}.
\] 
We conclude by observing that 
\begin{eqnarray*}
	|a_i(n)| 
	& \leq & 
	6 \sum_{0 \leq q \leq t_{\ell_2(n)-i+1}} |y_{\ell_2(n)-i+1,q}|
	+
	4 \sum_{0 \leq q \leq t_{\ell_2(n)-i}} |y_{\ell_2(n)-i,q}|	\\
	& \leq & 
	10 \cdot 2^i.
\end{eqnarray*}
Their combination yields a term $6\cdot 3^{\ell-1}$.  
\end{proof}

\begin{remark}
With a deeper analysis, one could probably refine the above lemma (even though this is not required for what remains). 
Let $(F(n))_{n\ge 0}$ be the Fibonacci sequence where $F(0)=1$, $F(1)=2$ and $F(n+2)=F(n+1)+F(n)$ for all $n\ge 0$. 
For all $\alpha\in[0,1]$ and all $i\ge 1$, we claim that $|a_i(\alpha)|\le 6\, F(i-1)$. 
The equality holds for $\alpha=1/3$. 
In this case, the sequence $(w_n(\alpha))_{n\ge 1}$ converges to $(10)^\omega$.
\end{remark}

\begin{proof}[Proof of Proposition~\ref{pro:limexi}.]
Using the $3$-decomposition of $A(e_n(\alpha))$, we have
$$\phi_n(\alpha) =\frac{1}{3^{\log_2(e_n(\alpha))}}
\sum_{i=0}^{\ell_2(e_n(\alpha))} a_i(e_n(\alpha))\, 3^{\ell_2(e_n(\alpha))-i}.
$$

Note that $\log_2(e_n(\alpha))=n+1+\{\log_2(e_n(\alpha))\}$.
Moreover $\ell_2(e_n(\alpha))=n$ if $e_n(\alpha)=2^{n+1}+r$ with $0\le r\le 2^n$, or $\ell_2(e_n(\alpha))=n+1$ if $e_n(\alpha)=2^{n+1}+r$ with $2^n< r< 2^{n+1}$.

If $\alpha<1/2$, then $e_n(\alpha)=2^{n+1}+r$ with $r\le 2^n-1$.
If $\alpha\ge 1/2$, then $e_n(\alpha)=2^{n+1}+r$ with $2^n+1\le r< 2^{n+1}$. 
Consequently, if $\alpha<1/2$, we have
\begin{equation}
\label{eq:phi_n<1/2}
\phi_n(\alpha)=\frac{1}{3^{1+\{\log_2(e_n(\alpha))\}}}\sum_{i=0}^{n} \frac{a_i(e_n(\alpha))}{3^{i}}
\end{equation}
and if $\alpha\ge 1/2$, we get
\begin{equation}
\label{eq:phi_n>1/2}
\phi_n(\alpha)=\frac{1}{3^{\{\log_2(e_n(\alpha))\}}}\sum_{i=0}^{n+1} \frac{a_i(e_n(\alpha))}{3^{i}}.
\end{equation}
 
First, in both expressions, the sums are converging when $n$ tends to infinity to the series
$$\sum_{i=0}^{+\infty} \frac{a_i(\alpha)}{3^{i}}.$$
Indeed, thanks to Lemma~\ref{lem:conv}, the sequence of finite words $(\mathsf{3dec}(A(e_n(\alpha))))_{n\ge 1}$ converges to $\mathbf{a}(\alpha)$. 
Moreover, due to Lemma~\ref{lem:estimate}, the sequence of partial sums uniformly converges to the series.

By Definition~\ref{def:22}, we get
$$
\left|\frac{e_n(\alpha)}{2^{n+1}}-(\alpha+1)\right| \le \frac{3}{2^{n+1}}.$$
Thus, the sequence of functions $(e_n(\alpha)/2^{n+1})_{n\ge 1}$ uniformly converges to $(\alpha + 1)$.
Since the function $\log_2$ is uniformly continuous on $[1,+\infty[$, the sequence $(\log_2 ( e_n(\alpha)/2^{n+1}))_{n\ge 1}$ also uniformly converges to $\log_2(\alpha+1)$. Now observe that 
$$
\log_2 \left( \dfrac{e_n(\alpha)}{2^{n+1}}\right) = \lfloor \log_2(e_n(\alpha))\rfloor+\{\log_2(e_n(\alpha))\}-n-1 = \{\log_2(e_n(\alpha))\}.
$$

Let $\epsilon>0$. 
For all $\alpha \ge 1/2$, we observe, using~\eqref{eq:phi_n>1/2}, that the inequality 
\begin{align*}
| \phi_n(\alpha) - \Phi(\alpha) | &\le \left| \sum_{i=0}^{n+1} \frac{a_i(e_n(\alpha))}{3^{i}} \right| \cdot \left| \frac{1}{3^{\{\log_2(e_n(\alpha))\}}} - \frac{1}{3^{\log_2(\alpha+1)}} \right|\\ &+ \left| \frac{1}{3^{\log_2(\alpha+1)}} \right| \cdot \left| \sum_{i=0}^{n+1} \frac{a_i(e_n(\alpha))}{3^{i}} - \sum_{i=0}^{+\infty} \frac{a_i(\alpha)}{3^{i}} \right|<\epsilon
\end{align*}
is valid for $n$ large enough. 
Indeed, the first inequality comes from~\eqref{eq:phi_n>1/2}. 
For the second inequality, we know that 
$$\left| \sum_{i=0}^{n+1} \frac{a_i(e_n(\alpha))}{3^{i}} \right| < C$$
where $C$ is a positive constant. 
Moreover, the sequence of functions $(\{\log_2(e_n(\alpha))\})_{n\ge 1}$ uniformly converges to $\log_2(\alpha+1)$ and thus
$$\left| \frac{1}{3^{\{\log_2(e_n(\alpha))\}}} - \frac{1}{3^{\log_2(\alpha+1)}} \right| < \frac{\epsilon}{2C} $$
for $n$ large enough. 
Finally, 
$$\left| \sum_{i=0}^{n+1} \frac{a_i(e_n(\alpha))}{3^{i}} - \sum_{i=0}^{+\infty} \frac{a_i(\alpha)}{3^{i}} \right|< \frac{\epsilon}{2}$$
for $n$ large enough. 
One proceeds similarly with~\eqref{eq:phi_n<1/2} for the case where $\alpha < 1/2$.

\end{proof}

The function $\Phi$ defined by Proposition~\ref{pro:limexi} takes particular values over rational numbers of the form $r/2^k$ with $r<2^k$ odd. 
This lemma is the key point to get an exact formula in Theorem~\ref{thm:convergence}.

\begin{lemma}\label{lem:rationalPhi}
    Let $k\ge 1$ and $r<2^k$ be integers. 
    We have
$$A(2^k+r)=3^{\log_2(2^k+r)}\, \Phi\left(\frac{r}{2^k}\right).$$
\end{lemma}

\begin{proof}
For $n\ge k$, we have 
$$w_n\left(\frac{r}{2^k}\right)=\rep_2(2^k+r)\, 0^{n-k}\, 1
\quad\text{ and }\quad
e_n\left(\frac{r}{2^k}\right)=2^{n-k+1}(2^k+r)+1.$$
By definition of $\Phi$, we know that 
$$\lim_{n\to+\infty} \frac{A(2^{n-k+1}(2^k+r)+1)}{3^{\log_2(2^{n-k+1}(2^k+r)+1)}}=\Phi\left(\frac{r}{2^k}\right).$$
Thanks to Corollary~\ref{cor:32n},  for all $n\ge k$, we have
\begin{equation}
    \label{eq:constant_sequence}
    \frac{A(2^{n-k+1}(2^k+r))}{3^{n-k+1}}=A(2^k+r).
\end{equation}
Now observe that 
$$\frac{A(2^{n-k+1}(2^k+r)+1)}{A(2^{n-k+1}(2^k+r))}
\frac{3^{\log_2(2^{n-k+1}(2^k+r))}}{3^{\log_2(2^{n-k+1}(2^k+r)+1)}}\to 1
$$
when $n$ tends to infinity because the first factor is equal to 
\begin{equation}
    \label{eq:arg1}
    1+\frac{\mathsf{s}(2^{n-k+1}(2^k+r)+1)}{A(2^{n-k+1}(2^k+r))},
\end{equation}
$A(2^{n-k+1}(2^k+r))\ge 3^n$ 
and, by Corollary~\ref{cor:majs}, $\mathsf{s}(m) \leq 2m$ for all $m$. 
This proves that the sequence 
$$\left(\frac{A(2^{n-k+1}(2^k+r))}{3^{\log_2(2^{n-k+1}(2^k+r))}} \right)_{n\ge k}$$
also converges to $\Phi(r/2^k)$. 
But from \eqref{eq:constant_sequence}, this sequence is constant and equal to 
$$\frac{A(2^k+r)}{3^{\log_2(2^k+r)}}.$$
\end{proof}

\begin{proof}[Proof of Theorem~\ref{thm:convergence}.] This proof is divided into four parts: the exact formula for the sequence $(A(N))_{N\ge 0}$, the fact that $\Phi(0)=1$, the limit $\lim_{\alpha\to 1^-}\Phi(\alpha)=1$ and the continuity of the function $\Phi$. 
 
$\bullet$ Every integer $N\ge 1$ can be uniquely written as $N=2^j(2^k+r)$ where $j\ge 0$ maximum, $k\ge 0$ and $r$ in $\{0,\ldots,2^k-1\}$. 
Thanks to Corollary~\ref{cor:32n}, $A(N)=3^j\, A(2^k+r)$. From Lemma~\ref{lem:rationalPhi}, we get
$$A(N)=3^j\,  A(2^k+r)=3^{j+\log_2(2^k+r)}\, \Phi\left(\frac{r}{2^k}\right).$$
To obtain the relation of the statement, observe that
$$\relpos_2(N)=\frac{N-2^{j+k}}{2^{j+k}}=\frac{r}{2^k}.$$

$\bullet$ From Lemma~\ref{lem:rationalPhi}, for any $k\ge 1$, we have
$$A(2^k)=3^{\log_2(2^k)}\Phi(0)$$
but, $A(2^k)=3^k$ thanks to Lemma~\ref{cor:32n}. 
Hence, $\Phi(0)=1$. 

$\bullet$ To show that $$\lim_{\alpha\to 1^-}\Phi(\alpha)= \lim_{\alpha\to 1^-} \lim_{n\to +\infty} \phi_n(\alpha) = 1,$$
we make use of the uniform convergence and permute the two limits 
$$ \lim_{n\to +\infty} \lim_{\alpha\to 1^-} \phi_n(\alpha) =  \lim_{n\to +\infty} \lim_{\alpha\to 1^-} \frac{1}{3^{\{\log_2(e_n(\alpha))\}}}\sum_{i=0}^{n+1} \frac{a_i(e_n(\alpha))}{3^{i}}. $$
Observe that if $\alpha$ is close enough to $1$, then the infinite word $\rep_2(\alpha)$ has a long prefix containing only letters $1$. By definition, we get $w_n(\alpha)=1^{n+2}$ and $e_n(\alpha)=2^{n+2}-1$. 
Iteratively applying Lemma~\ref{lem:recurrence} gives
\begin{align*}
A(2^{n+1}+2^{n+1}-1) &= 4\cdot 3^{n+1} - 2\cdot 3^{n} - A(2^{n}+1) - A(1) \\
&= 4\cdot 3^{n+1} - 2\cdot 3^{n} - 2\cdot 3^{n-1} - A(2^{n-1}+1) - 2\cdot A(1).
\end{align*}
This yields
$$ (a_i(e_n(\alpha)))_{0\le i \le n+1} = 4,-2,-2,-2,\ldots, -2, a_n(e_n(\alpha)), a_{n+1}(e_n(\alpha))$$
and since, by Lemma~\ref{lem:estimate}, the last two terms are respectively less than $10\cdot 2^n$ and $10\cdot 2^{n+1}$, the limit $\lim_{\alpha\to 1^-}\Phi(\alpha)$ is equal to 
$$\lim_{n\to +\infty}\frac{1}{3} \left(4-2 \sum_{i=1}^{n-1} 3^{-i}\right)= \lim _{n\to +\infty}\frac{1}{3} (3 + 3^{1-n})=1.$$

$\bullet$ We finally prove that $\Phi$ is continuous. 
Let $\alpha \in [0,1)$ and let us write $\rep_2(\alpha) = (d_n)_{n \geq 1}$. 
To show that $\Phi$ is continuous at $\alpha$, we make use of the uniform convergence of the sequence $(\phi_n)_{n \in \mathbb{N}}$ to $\Phi$ and consider
\begin{eqnarray*}
	\lim_{\gamma \to \alpha} 
	|\Phi(\gamma) - \Phi(\alpha)|
	&=& \lim_{\gamma \to \alpha} \lim_{n \to +\infty} 
	|\phi_n(\gamma)-\phi_n(\alpha)|	\\
	&=& \lim_{n \to +\infty} \lim_{\gamma \to \alpha}
	|\phi_n(\gamma)-\phi_n(\alpha)|
\end{eqnarray*}
First assume that $\alpha$ is not of the form $r/2^k$ with $k \geq 1$, $0 \leq r <2^k$, and $r$ odd, i.e., $(d_n)_{n \geq 1}$ does not belong to $\{0,1\}^*1 0^\omega$.
For any fixed integer $n$, we can choose $\gamma$ close enough to $\alpha$ such that $\rep_2(\gamma) \in d_1 d_2 \cdots d_n \{0,1\}^\omega$.
Therefore, we have $e_n(\gamma) = e_n(\alpha)$, hence $\phi_n(\gamma) = \phi_n(\alpha)$.

Now assume that $\rep_2(\alpha) = d_1d_2 \cdots d_k 0^\omega$ with $d_k = 1$.
For any fixed integer $n > k$, we can chose $\gamma$ close enough to $\alpha$ such that
\begin{eqnarray*}
	\rep_2(\gamma)  \in & d_1d_2 \cdots d_k 0^n \{0,1\}^\omega,
	& \text{if }  \gamma \geq \alpha; \\
	\rep_2(\gamma)  \in & d_1d_2 \cdots d_{k-1} 01^n \{0,1\}^\omega, & \text{if } \gamma < \alpha.
\end{eqnarray*}
If $\gamma \geq \alpha$, we get $\phi_n(\gamma) = \phi_n(\alpha)$ as in the first case.
If $\gamma < \alpha$, we get 
\begin{eqnarray*}
	e_n(\alpha) &=& 2^{n+1} + 2 \sum_{i=1}^k d_i 2^{n-i} +1; \\
	e_n(\gamma) &=& 2^{n+1} + 2 \sum_{i=1}^{k-1} d_i 2^{n-i} + 2 \sum_{j=0}^{n-k-1} 2^j + 1.
\end{eqnarray*}
We thus get $e_n(\alpha) = e_n(\gamma)+2$ and, since ${\lfloor\log_2(e_n(\alpha))\rfloor} = {\lfloor \log_2(e_n(\gamma))\rfloor} = n+1$, we get 
\begin{eqnarray*}
	|\phi_n(\alpha) - \phi_n(\gamma)| 
	&\leq &
	\left| \frac{A(e_n(\alpha))}{3^{n+1}} 
	\left( \frac{1}{3^{\{\log_2(e_n(\alpha))\}}} - \frac{1}{3^{\{\log_2(e_n(\gamma))\}}}
	\right) 
	\right| \\
	&+&
	\left|
	\frac{1}{3^{\log_2(e_n(\gamma))}} 
	\left( A(e_n(\alpha) - A(e_n(\gamma)) \right)
	\right|
\end{eqnarray*}
For the first term, the factor $A(e_n(\alpha))/3^{n+1}$ converges to the series $\sum_{i=0}^{+\infty} (a_i(\alpha)/3^i)$ when $n$ tends to infinity and the factor $(1/3^{\{\log_2(e_n(\alpha))\}} - 1/3^{\{\log_2(e_n(\gamma))\}})$ tends to 0 when $n$ tends to infinity because
\begin{eqnarray*}
	\{\log_2(e_n(\alpha))\} - \{\log_2(e_n(\gamma))\} 
	&=&
	\log_2(e_n(\alpha)) - \log_2(e_n(\gamma)) 	\\
	&=&
	\log_2 \left( \frac{e_n(\gamma)+2}{e_n(\gamma)} \right) \leq 
	\log_2 \left( 1+ \frac{1}{2^n}\right)
\end{eqnarray*}
For the second term, we have, by Corollary~\ref{cor:majs}, $A(e_n(\alpha)) - A(e_n(\gamma)) = \mathsf{s}(e_n(\alpha)) + \mathsf{s}(e_n(\alpha)-1) \leq 2^{n+4}$ and $3^{\log_2(e_n(\gamma))} \geq 3^{n+1}$. 
This shows that $\Phi$ is continuous.
\end{proof}

\begin{remark}
    As stated in \cite[Remark 9.2.2]{BR}, observe that since the $1$-periodic function $\Phi$ is continuous, then it is completely defined in the interval $[0,1]$ by the values taken on the dense set of points of the form $r/2^k$.
    \end{remark}


\section{Summatory function of a Fibonacci-regular sequence using an exotic numeration system}\label{sec:fibonacci}

In this section, we show how our method can be extended to sequences that do not exhibit a $k$-regular structure. 
Instead of the base-$2$ numeration system, we may use other systems such as the \emph{Zeckendorf numeration system} \cite{Zeck}, also called the \emph{Fibonacci numeration system}. 
Take the Fibonacci sequence $F=(F(n))_{n\ge 0}$ defined by $F(0)=1$, $F(1)=2$ and $F(n+2)=F(n+1)+F(n)$ for all $n\ge 0$. 
Any integer $n\ge 2$ can be written as $F(\ell)+r$ for some $\ell\ge 1$ and $0 \le r< F(\ell-1)$.
More precisely, for any integer $n\ge 1$, there exists $\ell \ge 1$ such that 
$$ 
n = \sum_{j=0}^{\ell -1} c_j\, F(j)
$$
where the $c_j$'s are non-negative integers in $\{0,1\}$, $c_{\ell-1}$ is non-zero and $c_j \cdot c_{j-1}=0$ for all $j\in\{1,\ldots,\ell-1\}$. 
The word $c_{\ell -1}\cdots c_0$ is called the \emph{normal $F$-representation} of $n$ and is denoted by $\rep_F(n)$. 
Said otherwise, the word $c_{\ell-1}\cdots c_0$ is the greedy $F$-expansion of $n$. 
We set $\rep_F(0)=\varepsilon$. 
Finally, we say that $\rep_F(\mathbb{N})=1\{01,0\}^*\cup\{\varepsilon\}$ is the \emph{language of the numeration}. 
If $d_{\ell -1}\cdots d_0$ is a word over the alphabet $\{0,1\}$, then we set
$$\val_F(d_{\ell -1}\cdots d_0):=\sum_{j=0}^{\ell -1} d_j\, F(j).$$

Compared to the sequence $(\mathrm{s}(n))_{n\ge 0}$, the sequence $(\mathrm{s}_F(n))_{n\ge 0}$ defined by~\eqref{eq:defSF} only takes into account words not containing two consecutive $1$'s. 
A major difference with the integer base case is that the sequence $(\mathsf{s}_F(n))_{n\ge 0}$ is not known to be $k$-regular for any $k\ge 2$. 
Thus we are not anymore in the setting of known results. 
Nevertheless, it is striking that we are still able to mimic the same strategy and obtain an expression for the summatory function 
$$A_F(n) := \sum_{j=0}^n \mathsf{s}_F(j).$$ 
The first few terms of $(A_F(n))_{n\ge 0}$ are
$$1, 3, 6, 10, 14, 19, 25, 31, 37, 45, 54, 62, 70, 77, 87, 99, 111, 123, 133, 145, \ldots$$
As in section~\ref{sec:base2et3}, we analogously consider a convenient $B$-decomposition of $A_F(N)$ based on the terms of a sequence $(B(n))_{n\ge 0}$. 

In~\cite{LRS2} we have showed that $(\mathsf{s}_F(n))_{n\ge 0}$ satisfies a recurrence relation of the same form as the one in Proposition~\ref{pro:rec}. 
This sequence is $F$-regular. 
This notion of regularity is a natural generalization of $2$-regular sequences to Fibonacci numeration system \cite{AST}. 

\begin{proposition}\label{pro:recF}
We have $\mathsf{s}_F(0)=1$, $\mathsf{s}_F(1)=2$ and, for all $\ell\ge 1$ and $0 \le r< F(\ell-1)$, 
$$
\mathsf{s}_F(F(\ell)+r)=
\begin{cases}
	\mathsf{s}_F(F(\ell-1)+r)+\mathsf{s}_F(r),
	& \text{ if }0\le r<F(\ell-2);\\
    2\mathsf{s}_F(r), 
    & \text{ if }F(\ell-2)\le r < F(\ell-1).
\end{cases}
$$
\end{proposition}

The following result is the analogue of Corollary~\ref{cor:majs} and is obtained by induction.

\begin{corollary}\label{cor:majsF}
For all $\ell \geq 1$ and all $0 \leq r < F(\ell-1)$, we have $\mathsf{s}_F(F(\ell)+r) \leq 2^{\ell+1}$.
\end{corollary}

\subsection{Preliminary results and introduction of a sequence $(B(n))_{n\ge 0}$}
For the base-$2$ numeration system, we had that $A(2^n)=3^n$ for all $n$ (see Lemma~\ref{lem:2^n 3^n}). 
We have a similar result in the Fibonacci case.

\begin{proposition}\label{pro:Bdom}
Let $(B(n))_{n \geq 0}$ be the sequence of integers defined by $B(0)=1$, $B(1)=3$, $B(2)=6$ and $B(n+3) = 2B(n+2)+B(n+1)-B(n)$ for all $n \geq 0$. 
For all $n\ge 0$, we have  
$$A_F(F(n)-1)= B(n).$$
\end{proposition}
\begin{proof}
The equalities $B(0)=A_F(F(0)-1)$, $B(1)=A_F(F(1)-1)$, $B(2)=A_F(F(2)-1)$ can be checked by hand.
Let us show that $A_F(F(n+3)-1) = 2A_F(F(n+2)-1)+A_F(F(n+1)-1)-A_F(F(n)-1)$ for all $n \geq 0$.
By definition, we have $A_F(F(n+3)-1) = \sum_{j=0}^{F(n+3)-1}\mathsf{s}_F(j)$.
The equality to prove is thus equivalent to 
\[
	\sum_{j=F(n+2)}^{F(n+3)-1} \mathsf{s}_F(j)
	=
	\sum_{j=0}^{F(n)-1} \mathsf{s}_F(j)
	+ 2 \sum_{j=F(n)}^{F(n+1)-1} \mathsf{s}_F(j)
	+ \sum_{j=F(n+1)}^{F(n+2)-1} \mathsf{s}_F(j).
\]
Observe that $\{j \mid F(n+2) \leq j < F(n+3)\} = \{F(n+2)+r \mid 0 \leq r < F(n+1)\}$.
We thus get, using Proposition~\ref{pro:recF},
\begin{eqnarray*}
	\sum_{j=F(n+2)}^{F(n+3)-1} \mathsf{s}_F(j) 
	&=& 	\sum_{r=0}^{F(n+1)-1} \mathsf{s}_F(F(n+2)+r)			\\
	&=& 	\sum_{r=0}^{F(n)-1} 	\mathsf{s}_F(F(n+2)+r)		
		+ \sum_{r=F(n)}^{F(n+1)-1} \mathsf{s}_F(F(n+2)+r)		\\
	&=& 	\sum_{r=0}^{F(n)-1} 	\left(\mathsf{s}_F(F(n+1)+r)+\mathsf{s}_F(r) \right)		
		+ \sum_{r=F(n)}^{F(n+1)-1} 2 \mathsf{s}_F(r)			\\	
	&=& 	\sum_{r=0}^{F(n)-1} \mathsf{s}_F(r)		
		+ 2 \sum_{r=F(n)}^{F(n+1)-1} \mathsf{s}_F(r)
		+ \sum_{r=0}^{F(n)-1} \mathsf{s}_F(F(n+1)+r).			
\end{eqnarray*}
We conclude by observing that $\{j \mid F(n+1) \leq j < F(n+2)\} = \{F(n+1)+r \mid 0 \leq r < F(n)\}$.
\end{proof}

The first few terms of the sequence $(B(n))_{n \geq 0}$ are
$$
1, 3, 6, 14, 31, 70, 157, 353, 793, 1782, 4004, 8997, 20216, 45425, 102069, 229347, \ldots
$$
The characteristic polynomial of the linear recurrence of $(B(n))_{n \geq 0}$ has three real roots as depicted in Figure~\ref{fig:poldeg}. 
\begin{figure}[h!t]
    \centering
    \scalebox{.4}{\includegraphics{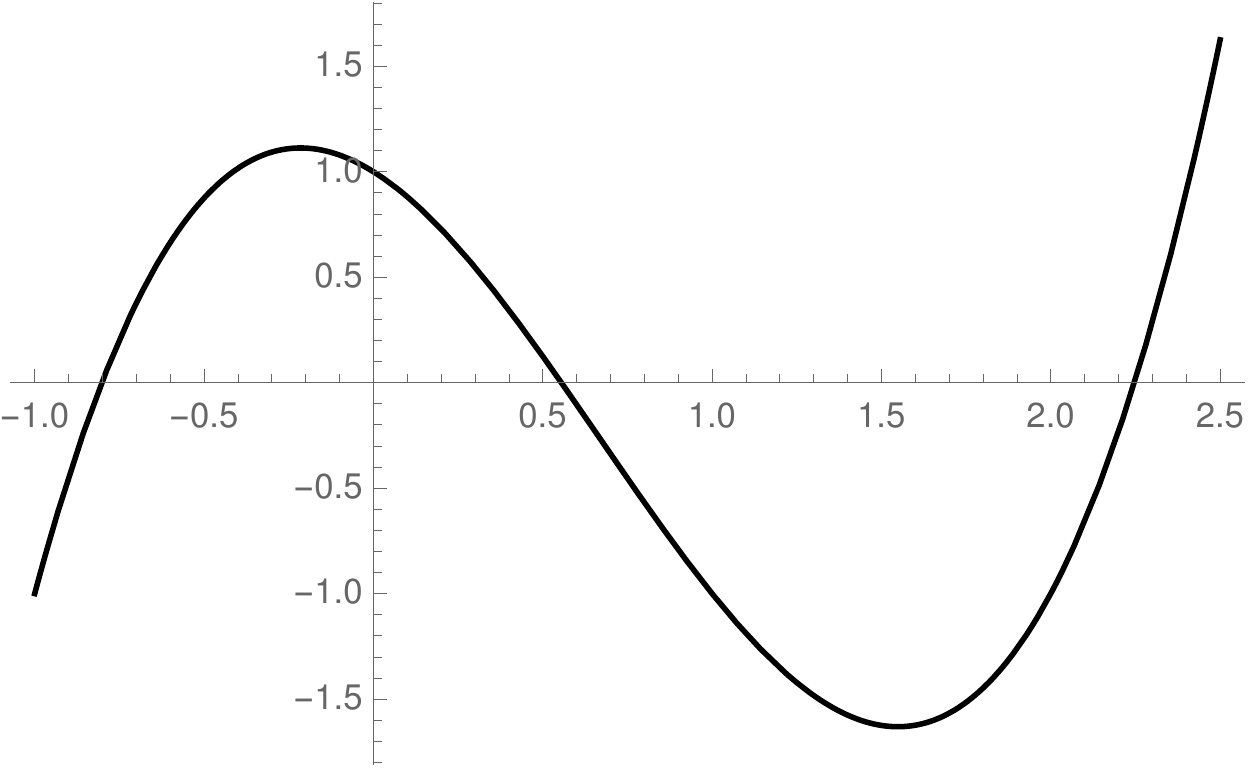}}
    \caption{The graph of $X^3-2 X^2 - X + 1$}
    \label{fig:poldeg}
\end{figure}
We let $\beta \approx 2.24698$ denote the root of maximal modulus of $X^3-2 X^2 - X + 1$. 
The other two roots are $\beta_2\approx -0.80194$ and $\beta_3\approx 0.55496$. 
From the classical theory of linear recurrences, 
there exist constants $c\approx 1.22041$, $c_2\approx -0.28011$ and $c_3\approx 0.0597$ such that, for all $n\in\mathbb{N}$,
\begin{equation}
    \label{eq:recB}
    B(n)=c\, \beta^n + c_2\, \beta_2^n+c_3\, \beta_3^n.
\end{equation}
In particular, we have
$$\lim_{n\to +\infty}\frac{B(n)}{c\, \beta^n}=1.$$
Thanks to Proposition~\ref{pro:Bdom}, we have an analogue of Lemma~\ref{lem:recurrence}.

\begin{lemma}\label{lem:recurrenceFib}
Let $\ell\ge 2$. 
If $0\le r< F(\ell-2)$, then
\begin{equation}
    \label{eq:partA}
    A_F(F(\ell)+r)= B(\ell)-B(\ell-1)+A_F(F(\ell-1)+r)+A_F(r).
\end{equation}
If $F(\ell-2)\le r < F(\ell-1)$, then
\begin{equation}
    \label{eq:partB}
    A_F(F(\ell)+r)= 2B(\ell)-B(\ell-1)-B(\ell-2)+2 A_F(r).
\end{equation}
\end{lemma}
\begin{proof}
Assume first that $0\le r < F(\ell-2)$. 
Applying Proposition~\ref{pro:recF} and Proposition~\ref{pro:Bdom}, we get
\begin{eqnarray*}
A_F(F(\ell)+r)&=&\sum_{j=0}^{F(\ell)-1} \mathsf{s}_F(j) + \sum_{j=0}^{r} \mathsf{s}_F(F(\ell)+j) \\
&=& B(\ell) + \sum_{j=0}^{r} \mathsf{s}_F(F(\ell-1)+j) + \sum_{j=0}^{r} \mathsf{s}_F(j) \\
&=& B(\ell) + \sum_{j=0}^{F(\ell-1)+r} \mathsf{s}_F(j) - \sum_{j=0}^{F(\ell-1)-1} \mathsf{s}_F(j) + A_F(r) \\
&=& B(\ell)+A_F(F(\ell-1)+r)-B(\ell-1)+A_F(r).
\end{eqnarray*}
Let us prove the second part of the result by assuming that $F(\ell-2)\le r < F(\ell-1)$. 
According to the second case of Proposition~\ref{pro:recF}, we have
\begin{eqnarray*}
A_F(F(\ell)+r)&=&\sum_{j=0}^{F(\ell)+F(\ell-2)-1} \mathsf{s}_F(j) + \sum_{j=F(\ell-2)}^{r} \mathsf{s}_F(F(\ell)+j) \\
&=& A_F(F(\ell)+F(\ell-2)-1) + 2 \sum_{j=F(\ell-2)}^{r} \mathsf{s}_F(j)\\
&=& A_F(F(\ell)+F(\ell-2)-1) + 2 \left( A_F(r) - \sum_{j=0}^{F(\ell-2)-1} \mathsf{s}_F(j) \right).
\end{eqnarray*}
Applying the first part of the result to the term $A_F(F(\ell)+F(\ell-2)-1)$, we obtain
\begin{multline*}
A_F(F(\ell)+r) = B(\ell) - B(\ell-1) + A_F(F(\ell-1)+F(\ell-2)-1) \\
+ A_F(F(\ell-2)-1) + 2 ( A_F(r) - A_F(F(\ell-2)-1) )
\end{multline*}
and next, with Proposition~\ref{pro:Bdom}, we get
\begin{eqnarray*}
A_F(F(\ell)+r)&=&  B(\ell) - B(\ell-1) + B(\ell) + B(\ell-2) + 2 A_F(r) - 2 B(\ell-2) \\
&=& 2 B(\ell) - B(\ell-1) - B(\ell-2) + 2 A_F(r).
\end{eqnarray*}
\end{proof}

\subsection{$B$-decomposition of $A_F(n)$}
Similarly to the $3$-decomposition of $A(n)$ considered in Section~\ref{sec:3dec}, we will consider what we call the $B$-decomposition of $A_F(n)$. 
The idea is to apply iteratively Lemma~\ref{lem:recurrenceFib} to derive a decomposition of $A_F(n)$ as a particular linear combination of terms of the sequence $(B(n))_{n\ge 0}$. 
Indeed, each application of Lemma~\ref{lem:recurrenceFib} provides a ``leading'' term of the form $B(\ell)$ or $2B(\ell)$ plus terms of smaller indices. 
In this context, we choose to set $A_F(0)=1\cdot B(0)$, $A_F(1)=3\cdot B(0)$ and $A_F(2)=6\cdot B(0)$.

\begin{definition}[$B$-decomposition]\label{def:B-decomp}
Let $n\ge 3$. 
Iteratively applying Lemma~\ref{lem:recurrenceFib} provides a unique decomposition of the form 
$$
A_F(n)=\sum_{i=0}^{\ell_F(n)} b_i(n)\, B(\ell_F(n)-i),
$$
where $b_i(n)$ are integers, $b_{0}(n)\neq 0$ and $\ell_F(n)=|\rep_F(n)|-1$. 
We say that the word 
$$
\mathsf{Bdec}(A_F(n)):=b_0(n) \cdots b_{\ell_F(n)}(n)
$$
is the {\em $B$-decomposition} of $A_F(n)$. Observe that when the integer $n$ is clear from the context, we simply write $b_i$ instead of $b_i(n)$.  For the sake of clarity, we will also write $(b_0(n), \ldots, b_{\ell_F(n)}(n))$.
\end{definition}

As an example, we get $$A_F(42)=B(7)+B(6)-B(5)+2B(4)-3B(1)+27B(0).$$
We have $\ell_F(42)=7$ and the $B$-decomposition of $A_F(42)$ is $(1,1,-1,2, 0,0,-3,27)$. Table~\ref{tab:Bdecomp} displays the $B$-decomposition of $A_F(3),A_F(4),\ldots$ As in the base-$2$ case, observe that the $B$-decomposition is only defined for the integers $(A_F(n))_{n\ge 0}$.
\begin{table}[h!tb]
$$\begin{array}{c|ccccc|r}
n & b_0(n) & b_1(n) & b_2(n) & b_3(n) & b_4(n) & A_F(n) \\
\hline
3& 1 & -1 & 7  & &  &1\times 6 - 1 \times 3 + 7 \times 1=10\\
4& 2 & -1 & 5  & & &2\times 6 - 1 \times 3 + 5 \times 1=14\\
5& 1 & 0 & -1 & 8& &1\times 14   - 1 \times 3 + 8\times 1=19 \\
6& 1 & 1 & -1 & 8& &1\times 14 + 1 \times 6 - 1 \times 3 + 8\times 1=25 \\
7& 2 & -1 & -1  & 12 &  &2\times 14 -1 \times 6 - 1 \times 3 + 12\times 1=31\\
8& 1 & 0 & 0 & -1 & 9  & 1 \times 31  -1 \times 3 + 9\times 1=37 \\
\vdots&  &  &  &  & &  \\
\end{array}$$
\caption{The $B$-decomposition of $A_F(3),A_F(4),\ldots$}
    \label{tab:Bdecomp}
\end{table}

\begin{remark}\label{rem:casesFib}
Assume that we want to develop $A_F(n)$ using only Lemma~\ref{lem:recurrenceFib}, i.e., to get the $B$-decomposition of $A_F(n)$. 
Only two cases may occur.
\begin{enumerate}[(i)]
\item If $\rep_F(n)=100u$, with $u \in 0^* \rep_F(\mathbb{N})$, then we apply the first part of Lemma~\ref{lem:recurrenceFib} and we are left with evaluations of $A_F$ at integers whose normal $F$-representations are shorter and given by $10u \text{ and } \rep_F(\val_F(u))$. 
\item If $\rep_2(n)=101u$, with $u \in \{\varepsilon\} \cup 0^+ \rep_F(\mathbb{N})$, then we apply the second part of Lemma~\ref{lem:recurrenceFib} and we are left with evaluations of $A_F$ at an integer whose normal $F$-representation is shorter and given by $1u$.
\end{enumerate}
\end{remark}

Lemma~\ref{lem:conv} is adapted in the following way. 

\begin{lemma}\label{lem:convFib}
For all finite words $u,v,v'\in \{0,1\}^*$ such that $1uv, 1uv'\in 1\{0,01\}^*$ and $|u|\ge 2$, the $B$-decompositions of $A_F(\val_F(1uv))$ and $A_F(\val_F(1uv'))$ share the same coefficients $b_0,\ldots,b_{|u|-2}$.
\end{lemma}
\begin{proof}
The proof is similar to the proof of Lemma~\ref{lem:conv} and directly follows from Lemma~\ref{lem:recurrenceFib}. 
\end{proof}
\begin{example}
Take $\rep_F(163)=1(000010)1001$ and $\rep_F(673)=1(000010)0010000$. 
If we compare the $B$-decompositions of $A_F(163)$ and $A_F(673)$, they share the same first five coefficients.
$$\begin{array}{r||ccccc|ccccccccc}
n&b_0 & b_1 & b_2 & b_3 & b_4& b_5& b_6 & b_7 & b_8 & b_9 & b_{10} & b_{11} & b_{12} & b_{13} \\
\hline
163&1& 0& 0& 1& -1& 9& -5& 5&10 & -10 & 80 \\
673&1& 0& 0& 1& -1& 4& 0& 5& -5& 15& 0& 0 & -20 & 180 \\
\end{array}$$
\end{example}

In base $2$, evaluation of $A$ at powers of $2$ is of particular importance. 
Here we evaluate $A_F$ at $F(n)-1$.
\begin{lemma}\label{lem:Bdecomp_2ex}
The sequence $(\mathsf{Bdec}(A_F(F(n))))_{n\ge 0}$ converges to $1000\cdots$ and the sequence $(\mathsf{Bdec}(A_F(F(n)-1)))_{n\ge 0}$ converges to $(g_n)_{n\ge 0}$ where $g_0=2$, $g_1=-1$, $g_2=3$ and, for all $n\ge 3$, $g_n=2 g_{n-2}$. 
In particular, we have 
$$
\sum_{i=0}^{+\infty} \frac{g_i}{\beta^i}=\beta,
$$
and for all $n$, $g_n = \frac{4-\sqrt{2}}{2} (\sqrt{2})^n + \frac{4+\sqrt{2}}{2} (-\sqrt{2})^n$.
\end{lemma}
\begin{proof}
Let us prove the first part of the result. 
We show that, for all $n\ge 3$,
\begin{equation}\label{eq:BdecompFib}
\mathsf{Bdec}(A_F(F(n))) = (1, \underbrace{0, \cdots, 0}_{n-2 \text{ times}}, -1, n+5)
\end{equation} 
We proceed by induction on $n\ge 3$. 
One can check by hand that the result holds for $n=3$ using Lemma~\ref{lem:recurrenceFib}. 
Thus consider $n\ge 3$ and suppose the results holds for all $m<n+1$. 
From Lemma~\ref{lem:recurrenceFib}, we have
$$
A_F(F(n+1)) = B(n+1) - B(n) + A_F(F(n)) + A_F(0).
$$
By induction hypothesis, we find that
$$
A_F(F(n+1)) = B(n+1) - B(n) + ( B(n) - B(1) + (n+5)\cdot B(0)) + B(0),
$$
which proves~\eqref{eq:BdecompFib}. 
The convergence of the sequence of finite words $(\mathsf{Bdec}(A_F(F(n))))_{n\ge 0}$ to the infinite word $1000\cdots$ easily follows. 

Let us prove the second part of the result. 
We show that, for all $n\ge 3$, 
\begin{equation}\label{eq:Bded de 10101010}
\mathsf{Bdec}(A_F(F(n)-1)) = 
\begin{cases}
(g_0, g_1, \ldots, g_{n-2}, x), & \text{if } n \text{ is odd};\\
(g_0, g_1, \ldots, g_{n-3}, y,z), & \text{if } n \text{ is even};
\end{cases}
\end{equation}
where $x,y,z$ are integers. 
We proceed again by induction on $n\ge 3$. 
One can check by hand that the result holds for $n\in\{3,4\}$ using Lemma~\ref{lem:recurrenceFib}. 
Thus consider $n\ge 4$ and suppose the results holds for all $m<n+1$. 
Suppose first that $n$ is even. 
By Lemma~\ref{lem:recurrenceFib}, we have
$$
A_F(F(n+1)-1) = 2B(n) - B(n-1) - B(n-2) + 2 A_F(F(n-1)-1).
$$
Using the induction hypothesis with $\mathsf{Bdec}(A_F(F(n-1)-1))=(g_0, g_1, \ldots, g_{n-3}, x)$, we get
\begin{align*}
A_F(F(n+1)-1)
=& 2B(n) - B(n-1) - B(n-2) + 4 B(n-2) - 2 B(n-3) \\ 
&+ 6 B(n-4)+ \sum_{j=3}^{n-3} 2g_j B(n-j-2) + 2x B(0).
\end{align*}
By definition of the sequence $(g_n)_{n\ge 0}$, we have $2g_j=g_{j+2}$ and we finally obtain
\begin{align*}
A_F(F(n+1)-1)
=& 2B(n) - B(n-1) + 3B(n-2) + \sum_{j=3}^{n-1} g_{j} B(n-j) + 2x B(0),
\end{align*}
which concludes the case where $n$ is even. 
The case where $n$ is odd can be proved using the same argument. 

Let us prove the last part of the result. 
Using the definition of the sequence $(g_n)_{n\ge 0}$, we get
$$
\sum_{i=1}^{+\infty} \frac{g_i}{\beta^i} = \frac{-1}{\beta} + \frac{3}{\beta^2} + \frac{2}{\beta^2} \cdot \sum_{i=1}^{+\infty} \frac{g_i}{\beta^i},
\quad \text{ that is, } \quad
 \sum_{i=1}^{+\infty} \frac{g_i}{\beta^i} = \frac{-\beta + 3}{\beta^2 - 2}.
$$
Hence, since $\beta^3 - 2\beta^2 - \beta + 1 =0$, we have
$$
\sum_{i=0}^{+\infty} \frac{g_i}{\beta^i} = 2 + \frac{3-\beta}{\beta^2 - 2} = \frac{2\beta^2 - \beta -1}{\beta^2-2} = \beta.
$$
The equality $$g_n = \frac{4-\sqrt{2}}{2} (\sqrt{2})^n + \frac{4+\sqrt{2}}{2} (-\sqrt{2})^n$$ directly follows from the recurrence equation and the initial conditions defining the sequence $(g_n)_{n \geq 0}$.
\end{proof}

\subsection{Behavior of the sequence $(A_F(N))_{N\ge 0}$}

Let $\varphi$ be the golden ratio. 
In the following, we recall the notion of $\varphi$-expansion of a real number in $[0,1)$; for more on this subject and on numeration systems, see, for instance, \cite[Chap.~7]{Lot2}. 
The $\varphi$-expansion of $\alpha\in [0,1)$, denoted by $\rep_\varphi(\alpha)$, is  
the infinite word $d_1d_2d_3\cdots$ satisfying $\sum_{i\ge 1} d_i \varphi^{-i}=\alpha$ and for all $j \geq 1$, 
\begin{equation}\label{eq:maj_sum_phi}
\sum_{i\ge j} d_i \varphi^{-i} < \varphi^{-j+1}.
\end{equation}
Observe that $d_id_{i+1}\neq 11$ for all $i\ge 1$. 
The idea in the next definitions is that $\alpha$ gives the relative position of an integer in the interval $[F(n),F(n+1))$.

\begin{definition}
    Let $\alpha$ be a real number in $[0,1)$. 
    Define the sequence of finite words $(w_n(\alpha))_{n\ge 1}$ where $w_n(\alpha)$ is the prefix of length $n$ of the infinite word $10\rep_\varphi(\alpha)$. 
\end{definition}

\begin{definition}\label{def:e_n_Fib} Let $\alpha$ be a real number in $[0,1)$. 
For each $n\ge 1$, let us define
$$
e_n(\alpha)=\val_F(w_n(\alpha)).
$$
\end{definition}

Note that, since the $\varphi$-expansion of $\alpha$ does not contain any factor of the form $11$, the word $w_n(\alpha)$ is the normal $F$-representation of the integer $e_n(\alpha)$ belonging to the interval $[F(n-1), F(n))$. 

\begin{definition}
For each $n\ge 1$, we compute the $B$-decomposition of  $A_F(e_n(\alpha))$. 
We thus have a sequence of finite words $(b_0(e_n(\alpha))\cdots b_{n-1}(e_n(\alpha)))_{n\ge 1}$. 
Thanks to Lemma~\ref{lem:convFib}, this sequence of finite words converges to an infinite sequence of integers denoted by $\mathbf{b}(\alpha)=b_0(\alpha)\, b_1(\alpha)\, \cdots$.
\end{definition}

\begin{example} Take $\alpha=\pi-3$. 
The first few letters of $\rep_\varphi(\alpha)$ are
$$00001010100100010101 \cdots .$$
Thus, the first few terms of the sequence of finite words $(w_n(\alpha))_{n\ge 1}$ are
$$1, 10, 100, 1000, 10000, 100000, 1000001, 10000010, 100000101, 1000001010, \ldots$$
We get that the first few terms of $(e_n(\alpha)))_{n\ge 1}$ are $1,2,3,5,8,13, 22, 36, 59, 96, \ldots$  
$$
    \begin{array}{c|c||ccccccccccc}
n & e_n(\alpha) & b_0 & b_1 & b_2 & b_3 &\cdots \\
\hline
1 & 1 & 3 &  &  &  &  &  &  &  &  
    \\ 
 2 & 2 & 6 &  &  &  &  &  &  &  &    
    \\
 3 & 3 & 1 & -1 & 7 &  &  &  &  &  &    
    \\
 4 & 5 & 1 & 0 & -1 & 8 & & &  &  &  &     
   \\
 5 & 8 & 1 & 0 & 0 & -1 & 9 & &  &  &  &     \\
 6 & 13 & 1 & 0 & 0 & 0 & -1 & 10 & &  &  &     \\
 7 & 22 & 1 & 0 & 0 & 0 & 1 & -1 &  17& & &     \\
 8 & 36 & 1 & 0 & 0 & 0 & 1 & -1 &  -1 & 36 & &    \\
 9 & 59 & 1 & 0 & 0 & 0 & 1 & -1 & 11 & -6 & 30 &    \\
 10 & 96 & 1 & 0 & 0 & 0 & 1 & -1 & 11 & -6 & -6 & 72     \\
\end{array}$$
By computing the $B$-decomposition of $(A_F(e_n(\alpha)))_{n\ge 1}$, the first terms of the sequence $\mathbf{b}(\alpha)$ are $1, 0, 0, 0, 1, -1, 11, -6$.
\end{example}

To ensure convergence, a rough estimate is enough.

\begin{lemma}\label{lem:estimateFib}
For all $n\ge 3$ and all $0\le i\le \ell_F(n)$, we have 
$$|b_i(n)|\le 6\cdot 2^i.$$
In particular, for all $\alpha\in[0,1)$ and all $i\ge 0$, we have
$$|b_i(\alpha)|\le 6\cdot 2^i.$$
\end{lemma}

\begin{proof}
The proof follows the same lines as the proof of Lemma~\ref{lem:estimate}.
Let us write $n = F(\ell) + r$ with $\ell \geq 2$ and $0 \leq r < F(\ell-1)$. 
Using Definition~\ref{def:B-decomp}, let us write
$$
A_F(n)=\sum_{j=0}^{\ell} b_j(n)\, B(\ell-j),
$$
where $b_j(n)$ are integers, $b_{0}(n)\neq 0$. 
Let us fix some $i \in \{0,1,\dots,\ell\}$.
By Lemma~\ref{lem:recurrenceFib}, terms of the form 
\begin{equation}\label{eq:formFib}
\begin{array}{ll}
A_F(F(\ell-i)+r'), 		
& \text{where } r'\in\{0,\ldots,F(\ell-i-1)-1\}, \text{ or }\\
A_F(F(\ell-i+1)+r''), 	
& \text{where } r''\in\{0,\ldots,F(\ell-i)-1\}, \text{ or }\\
A_F(F(\ell-i+2)+r'''), 
& \text{where } r'''\in\{F(\ell-i),\ldots,F(\ell-i+1)-1\},
\end{array}
\end{equation}
are the only ones possibly contributing to $b_i(n)$. 

Terms of the first form give either $B(\ell-i)$ or $2B(\ell-i)$, depending on whether $0 \leq r' < F(\ell-i-2)$ or $F(\ell-i-2) \leq r' < F(\ell-i-1)$ respectively.

Terms of the second form with $F(\ell-i-1) \leq r'' <F(\ell-i)$ gives $-B(\ell-i)$ with one application of the lemma. 
If $F(\ell-i-2) \leq r'' <F(\ell-i-1)$, a first application of the lemma gives $-B(\ell-i)$ and the term $A_F(F(\ell-i)+r)$ which is of the first form. 
A second application of the lemma then gives $2B(\ell-i)$ and so the final contribution is $B(\ell-i)$.
Similarly, if $0 \leq r'' <F(\ell-i-2)$, the contributions given by the two applications of the lemma cancel each other out.

Like for terms of the second form, terms of the third form need two applications of the lemma because the first application gives a term $2 A_F(r''') = 2A_F(F(\ell-i)+r'''')$. 
The final contribution is then either $B(\ell-i)$ or $3B(\ell-i)$, depending on whether $0 \leq r'''' < F(\ell-i-2)$ or $F(\ell-i-2) \leq r'''' < F(\ell-i-1)$, respectively.

Mimicking the proof of Lemma~\ref{lem:estimate}, iterating Lemma~\ref{lem:recurrenceFib} on $A_F(F(\ell)+r)$ gives a linear combination of the form 
\[
	\sum_{j= \ell-i+1}^{\ell} y_j B(j) 
	+
	\sum_{j=0}^{\ell-i+2} 
	x_j A_F(F(j)+r'_j),
\]
where 
\[
	\sum_{j=0}^{\ell-i+2} |x_j|
	\leq 2^i.
\] 
We conclude by observing that 
\begin{eqnarray*}
	|b_i(n)| 
	& \leq & 
	2 |x_{\ell-i}|
	+
	|x_{\ell-i+1}|	
	+ 
	3 |x_{\ell-i+2}| \\
	& \leq & 
	6 \cdot 2^i.
\end{eqnarray*}
\end{proof}

Let $n$ be an integer such that $\rep_F(n)=10r_1\cdots r_k$ with $k \geq 1$ and $r_i\in\{0,1\}$ for all $i$. 
We define
$$\relpos_F(n)=\sum_{i=1}^k\frac{r_i}{\varphi^i} \quad \text{ and } \quad \log_F(n) = |\rep_F(n)|-1+\relpos_F(n).$$
Observe that $\relpos_F(e_n(\alpha))\to\alpha$ as $n\to +\infty$ and $\lfloor \log_F n \rfloor = |\rep_F(n)|-1$. 
As in the case of the base-$2$ expansions, we will introduce an auxiliary function $\Psi(\alpha)$, for $\alpha\in[0,1)$, defined as the limit of a converging sequence built on the $B$-decomposition of $A_F(e_n(\alpha))$. 
For all $n\ge 1$, let $\psi_n$ be the function defined, for $\alpha\in[0,1)$, by 
$$
\psi_n(\alpha) = \frac{A_F(e_n(\alpha))}{c\, \beta^{\log_F(e_n(\alpha))}},
$$
where $\beta$ and $c$ come from~\eqref{eq:recB}.
We have depicted the first functions $\psi_3,\ldots,\psi_{11}$ in Figure~\ref{fig:phin2}. 
For instance, $\psi_3$ is a step function built on two subintervals because $w_3(\alpha)$ can only takes two values: $100$ and $101$. 
In general, $w_n(\alpha)$ takes $F(n-2)$ values.
\begin{figure}[h!tb]
    \centering
    \scalebox{.29}{\includegraphics{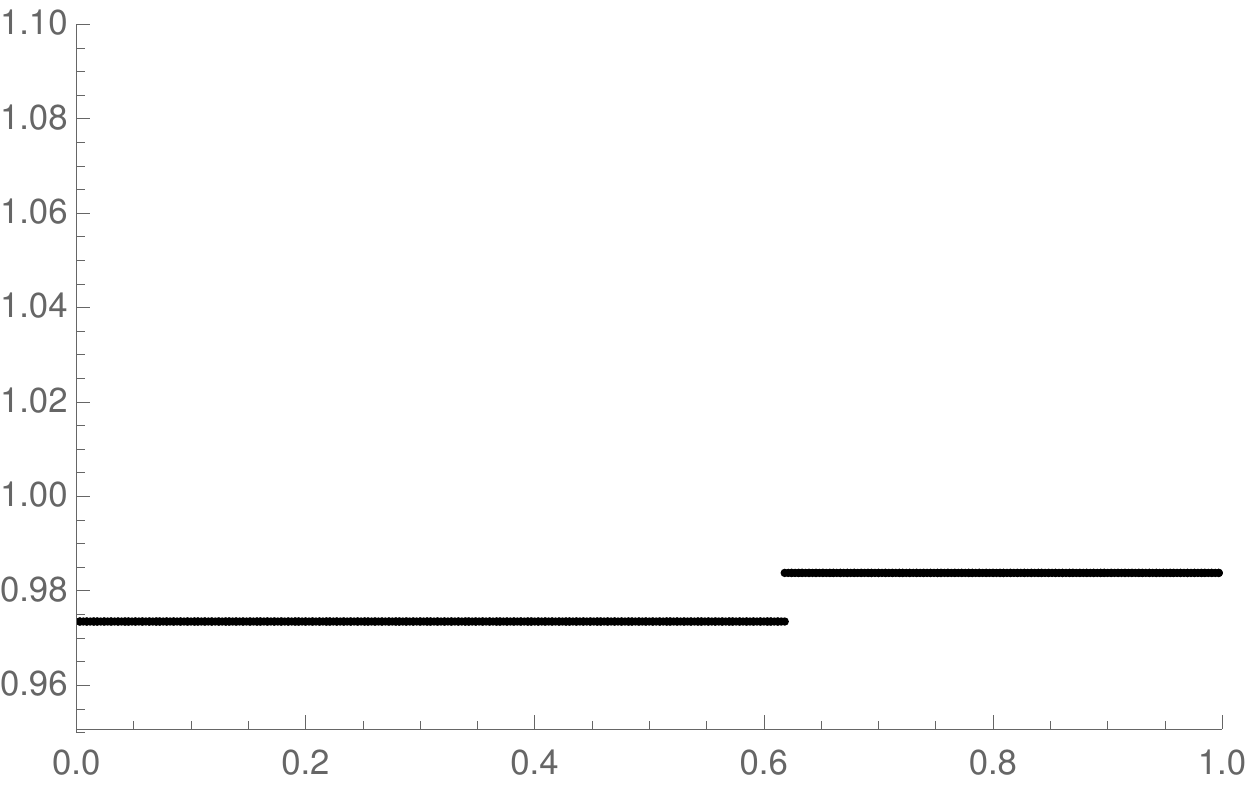}}\quad \scalebox{.29}{\includegraphics{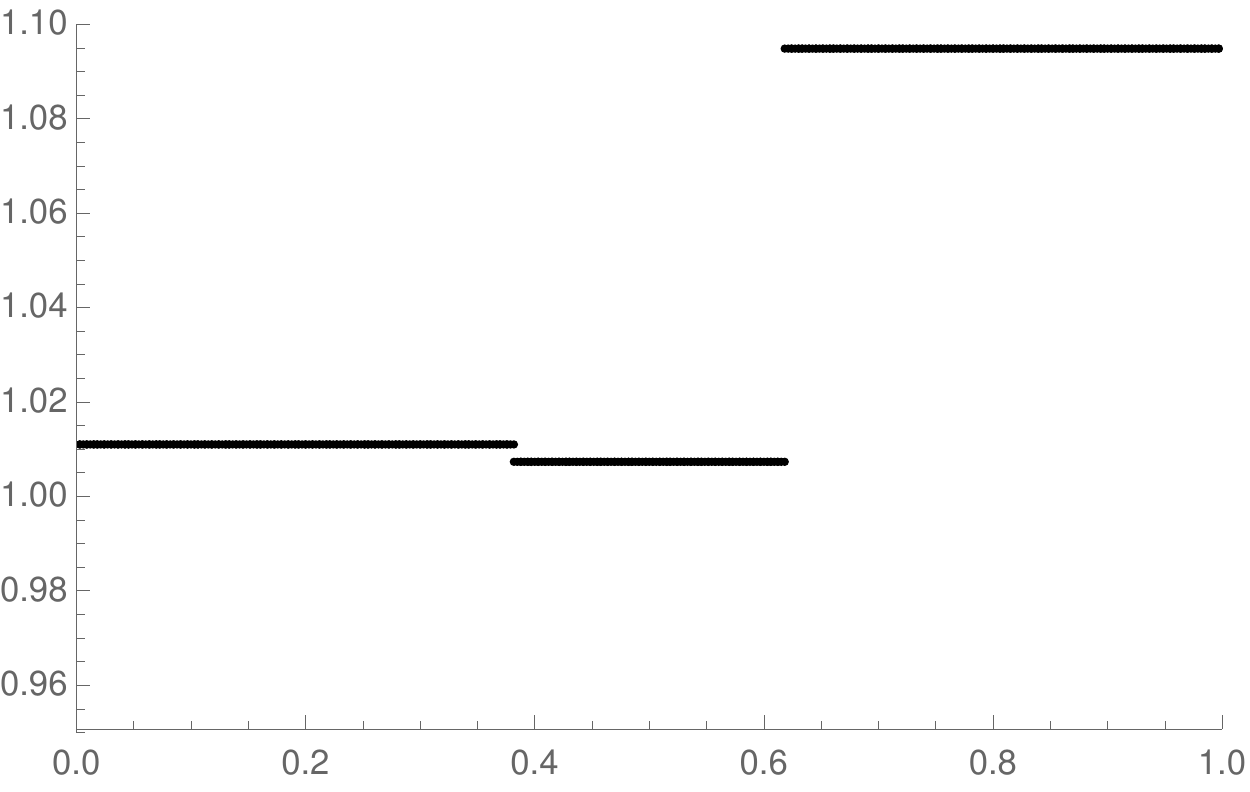}}\quad \scalebox{.29}{\includegraphics{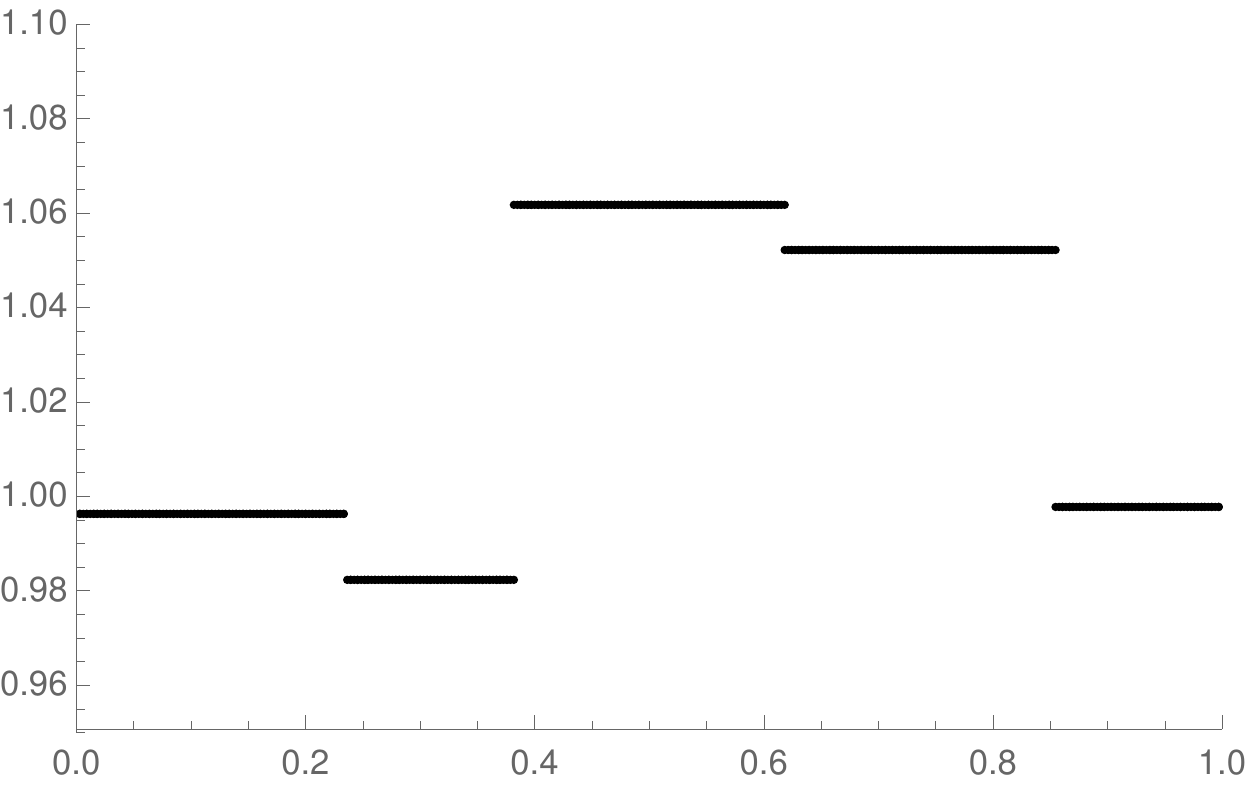}}\\

\scalebox{.29}{\includegraphics{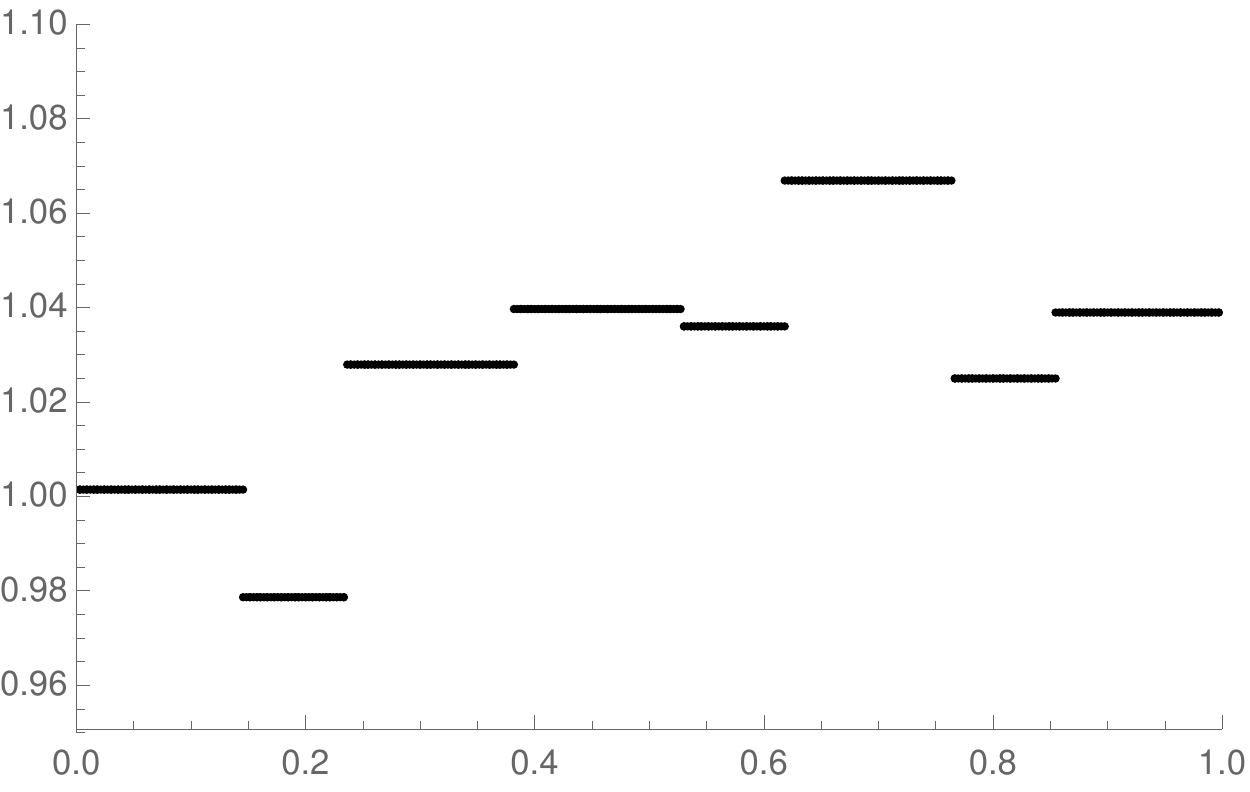}}\quad \scalebox{.29}{\includegraphics{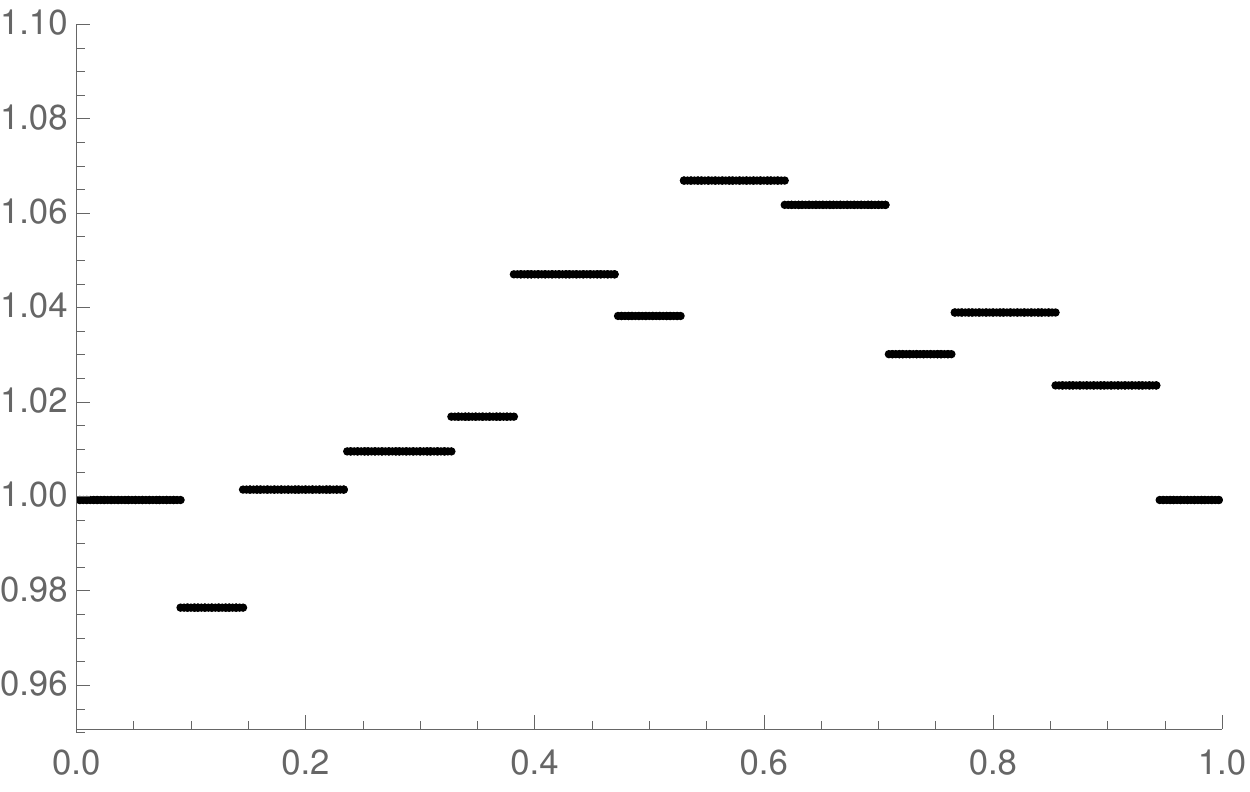}}\quad \scalebox{.29}{\includegraphics{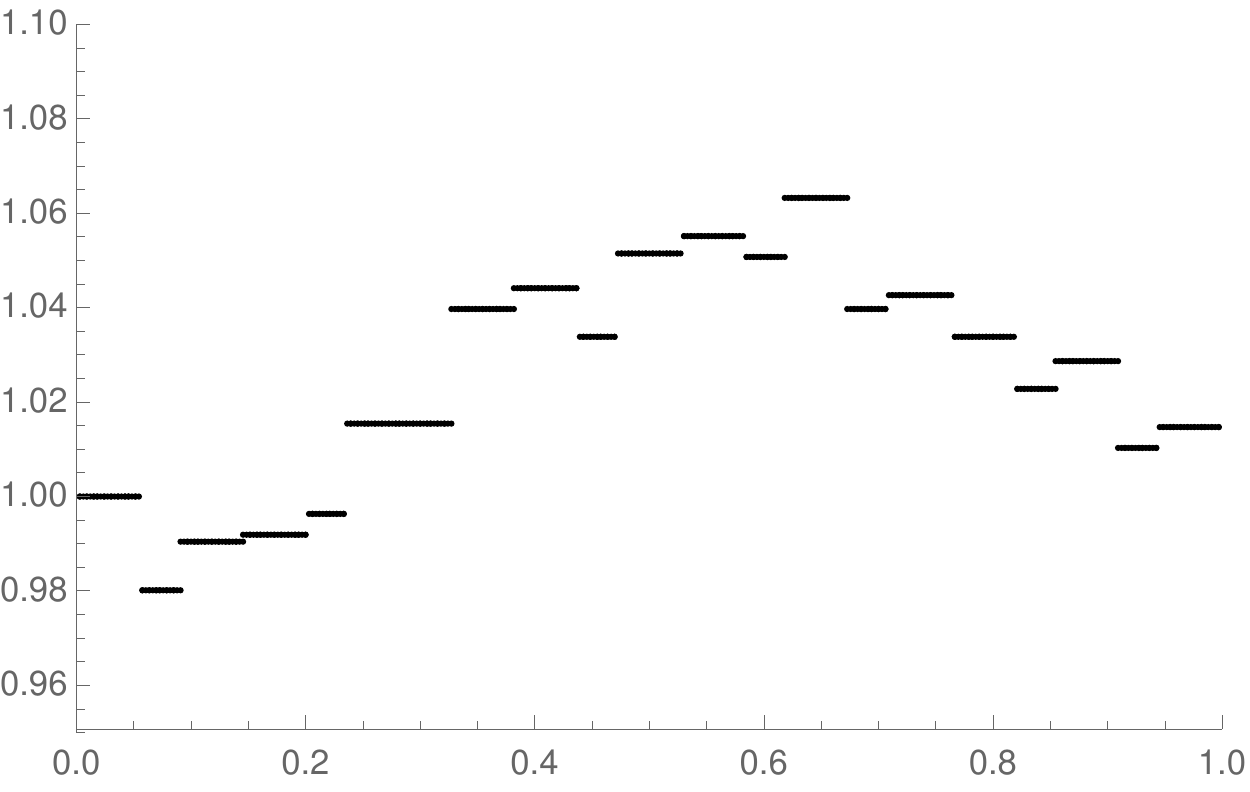}}\\

\scalebox{.29}{\includegraphics{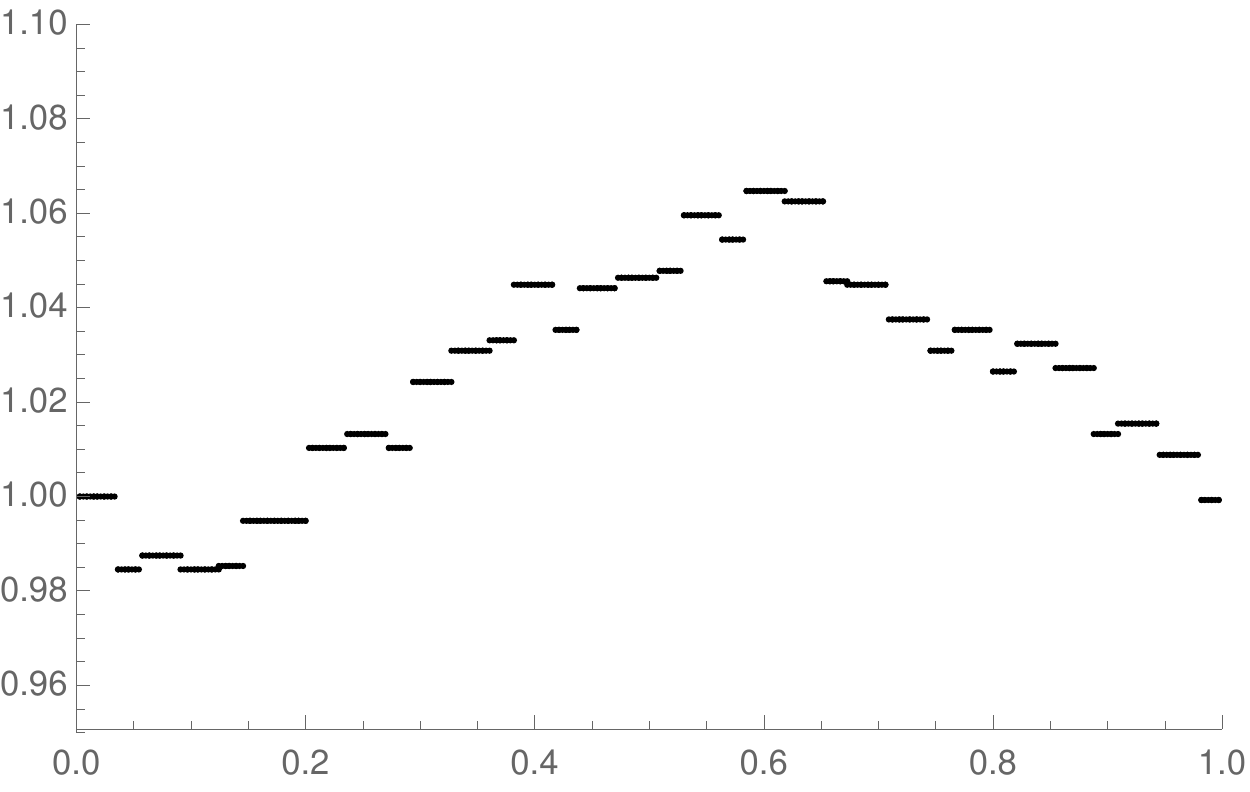}}\quad \scalebox{.29}{\includegraphics{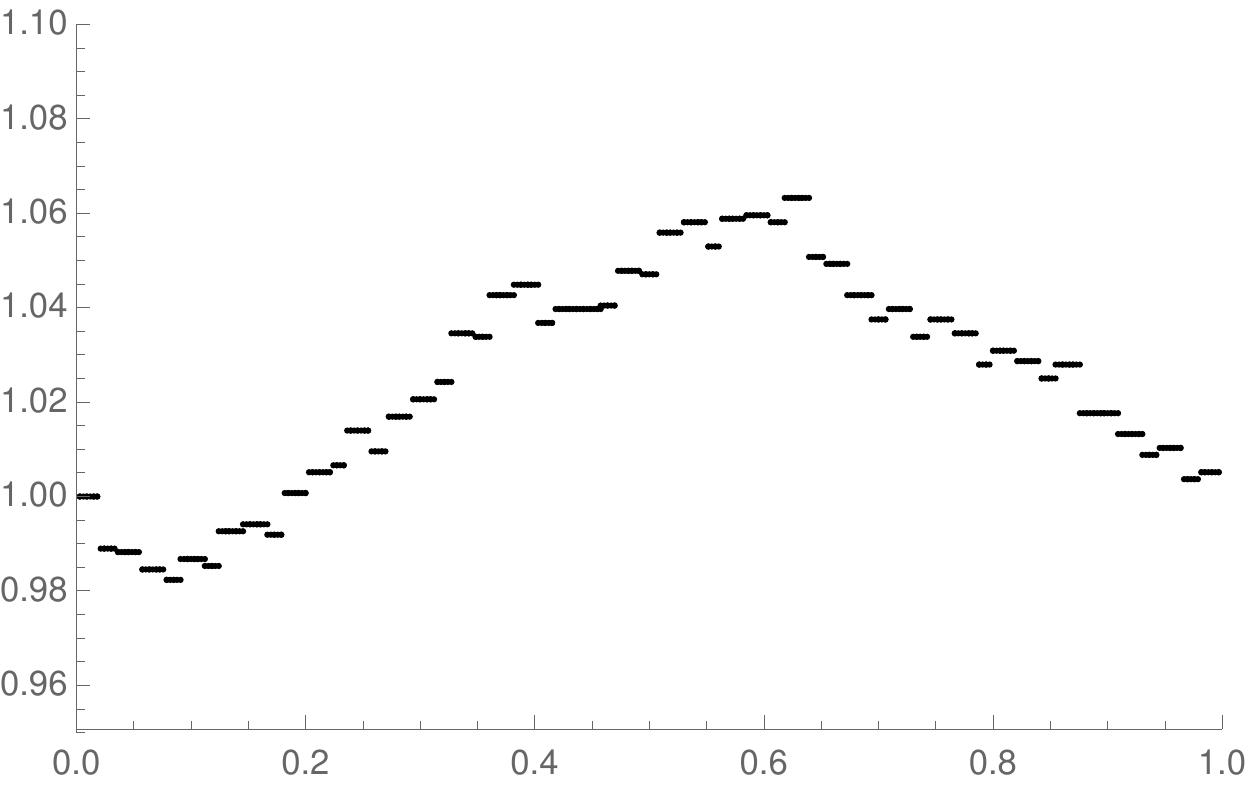}}\quad \scalebox{.29}{\includegraphics{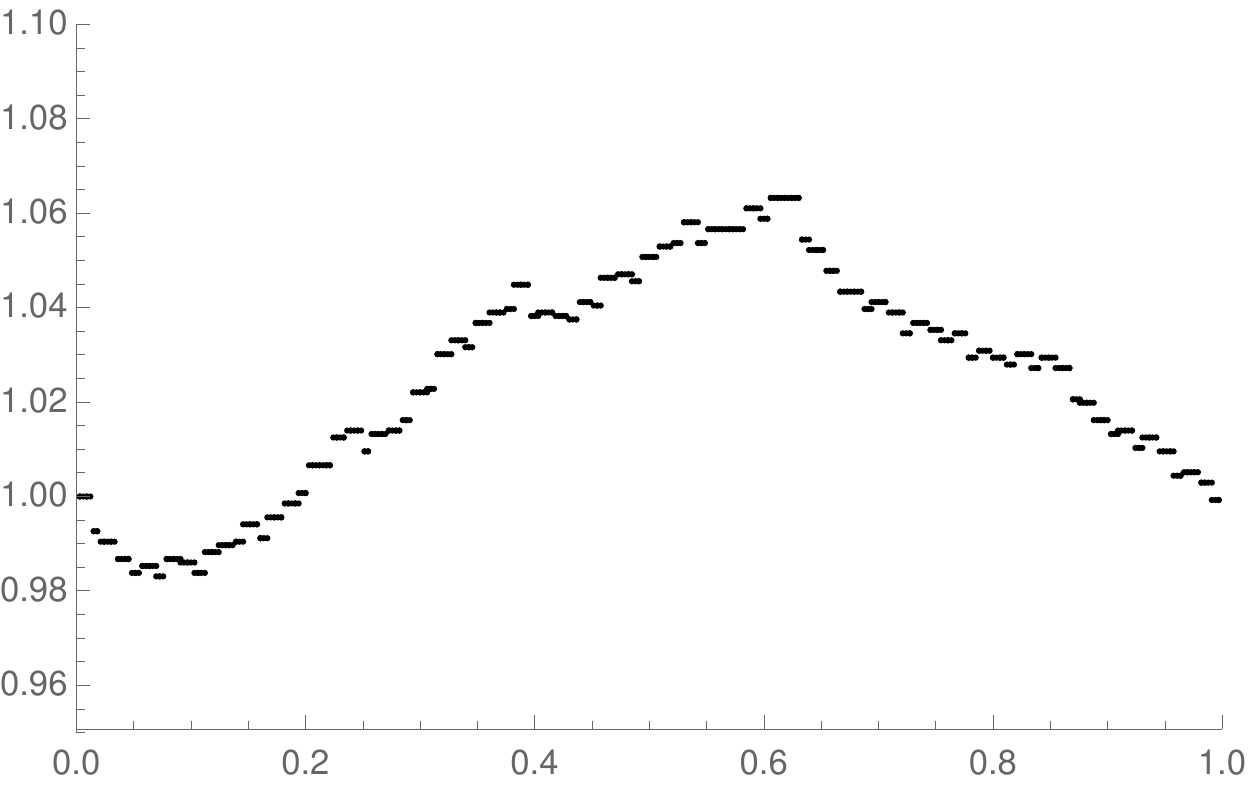}}
    \caption{Representation of $\psi_3,\ldots,\psi_{11}$.}
    \label{fig:phin2}
\end{figure}

\begin{proposition}\label{pro:existence_limit}
The sequence $(\psi_n)_{n\ge 1}$ uniformly converges to the function $\Psi$ defined for $\alpha\in[0,1)$ by
$$
\Psi(\alpha) := \dfrac{1}{\beta^{\alpha}} \sum\limits_{i=0}^{+\infty} \dfrac{b_i(\alpha)}{\beta^i}.
$$
\end{proposition}
\begin{proof}
Using the $B$-decomposition of $A_F(e_n(\alpha))$, we have, with $\ell_F(n) = n-1$,
\begin{align}
\label{eq:psi_n}
\psi_n(\alpha)=
\frac{A_F(e_n(\alpha))}{c\beta^{\log_F(e_n(\alpha))}}
&=\frac{1}{\beta^{\relpos_F(e_n(\alpha))}}\sum_{i=0}^{n-1} b_i(e_n(\alpha))\, \frac{B(n-1-i)}{c\beta^{n-1}}.
\end{align}
 
Firstly, the sum is converging when $n$ tends to infinity to the convergent series
$$\sum_{i=0}^{+\infty} \frac{b_i(\alpha)}{\beta^{i}}.$$
Indeed, thanks to Lemma~\ref{lem:convFib}, the sequence of finite words $(\mathsf{Bdec}(A_F(e_n(\alpha))))_{n\ge 1}$ converges to $\mathbf{b}(\alpha)$. 
Moreover, due to Lemma~\ref{lem:estimateFib} and \eqref{eq:recB}, the sequence of partial sums uniformly converges to the series.

Secondly, the sequence of functions $(\relpos_F(e_n(\alpha)))_{n\ge 1}$ is uniformly convergent. 
Indeed, if $\rep_\varphi(\alpha) = d_1d_2d_3\cdots$, then using~\eqref{eq:maj_sum_phi}, we have
\begin{equation}
\label{eq:diff relp e_n alpha}
\left|  \relpos_F(e_n(\alpha)) - \alpha  \right|
= \left| \sum_{i=1}^{n-2}\frac{d_i}{\varphi^i} - \sum_{i= 1}^{+\infty} \frac{d_i}{\varphi^{i}} \right| < \frac{1}{\varphi^{n-2}}.
\end{equation}

Let $\epsilon>0$. 
To conclude with the proof, we use the same reasoning as in the proof of Proposition~\ref{pro:limexi}. 
Indeed, using~\eqref{eq:psi_n}, the inequality
\begin{align*}
| \psi_n(\alpha) - \Psi(\alpha) | &\le \left| \frac{1}{\beta^{\relpos_F(e_n(\alpha))}} \right| \cdot \left| \sum_{i=0}^{n-1} b_i(e_n(\alpha)) \frac{B(n-1-i)}{c\beta^{n-1}} - \sum_{i=0}^{+\infty} \frac{b_i(\alpha)}{\beta^{i}} \right|\\ &+ \left| \sum_{i=0}^{+\infty} \frac{b_i(\alpha)}{\beta^{i}} \right| \cdot \left| \frac{1}{\beta^{\relpos_F(e_n(\alpha))}} - \frac{1}{\beta^{\alpha}} \right|<\epsilon
\end{align*}
holds for all $\alpha\in[0,1)$ and $n$ large enough. 
\end{proof}

Instead of considering rational numbers of the form $r/2^k$, we use the set 
$$D:=\left\{ \sum_{i=1}^k \frac{r_i}{\varphi^i} \mid k\ge 1,\ r_1\cdots r_k \in\{1,\varepsilon\}\{0,01\}^*\right\},$$
which is dense in $[0,1]$. 
The next result makes explicit the values taken by $\Psi$ on the set $D$.
\begin{lemma}\label{lem:equiv_rationnel}
    Let $r_1\cdots r_k \in\{1,\varepsilon\}\{0,01\}^*$ with $k \geq 1$ and $\alpha=\sum_{i=1}^k {r_i}/{\varphi^i}$. 
    We have
$$\Psi\left(\alpha\right)
=\sum_{i=0}^{k-1} \frac{b_{i}(m)}{\beta^{i+\alpha}} + \frac{b_{k}(\alpha)}{\beta^{k+\alpha}} + \frac{b_{k+1}(\alpha)}{\beta^{k+1+\alpha}} $$
where  $m=\val_F(10r_1\cdots r_k)$ and $b_0(m)\cdots b_{k+1}(m)$ is the $B$-decomposition of $A_F(m)$.
\end{lemma}
\begin{proof} 
We have $10\rep_\varphi(\alpha)=10r_1\cdots r_k 0^\omega$ and $w_n(\alpha)$ is the prefix of length $n$ of this word. 
For large enough $n$, due to Lemma~\ref{lem:convFib}, $\mathsf{Bdec}(A_F(e_n(\alpha)))$ has a prefix equal to $b_0(m)\cdots b_{k-1}(m)$. 
More precisely, it is of the form 
$$b_0(m)\cdots b_{k-1}(m)\, b_{k}(e_n(\alpha))\, b_{k+1}(e_n(\alpha)) \, 0^{n-k-4}\, b_{n-2}(e_n(\alpha))\, b_{n-1}(e_n(\alpha)).$$
This is again a consequence of Lemma~\ref{lem:recurrenceFib}.
Applying recursively this lemma to $A_F(e_n(\alpha))$, we will be left with the evaluation of $A_F(F(n-k-2))$. 
As in the proof of Lemma~\ref{lem:Bdecomp_2ex} with \eqref{eq:BdecompFib}, $A_F(F(n-k-2))$ yields $B(n-k-2)-B(1)+(n-k+3)B(0)$ explaining the block of zeroes.

Due to Proposition~\ref{pro:existence_limit}, we get
\begin{align*}
\Psi(\alpha)=\lim_{n\to +\infty} & \frac{1}{c\, \beta^{\log_F(e_n(\alpha))}}
\biggl( 
 \sum_{i=0}^{k-1} b_{i}(m) B(n-1-i)  \\
 &+
 b_{k}(e_n(\alpha)) B(n-1-k) + b_{k+1}(e_n(\alpha)) B(n-2-k)\\
 & + b_{n-2}(e_n(\alpha)) B(1) + b_{n-1}(e_n(\alpha))B(0)
\biggr).
\end{align*}
Using~\eqref{eq:recB}, we get
$$\lim_{n\to +\infty} \frac{B(n-1-i)}{c\, \beta^{\log_F(e_n(\alpha))}}=\frac{1}{\beta^{i+\alpha}}$$
and, for $j=1,2$, 
$$\lim_{n\to +\infty} \frac{b_{k}(e_n(\alpha))B(n-j-k)}{c\, \beta^{\log_F(e_n(\alpha))}}= \frac{b_{k}(\alpha)}{\beta^{k+j-1+\alpha}}.$$
Now, 
$$\lim_{n\to +\infty} \frac{b_{n-2}(e_n(\alpha))B(1)}{c\, \beta^{\log_F(e_n(\alpha))}}=0 \quad \text{ 
and }\quad \lim_{n\to +\infty} \frac{b_{n-1}(e_n(\alpha))B(0)}{c\, \beta^{\log_F(e_n(\alpha))}}=0$$
since $|b_{n-2}(e_n(\alpha))| \le 6\cdot 2^{n-2}$ and $|b_{n-1}(e_n(\alpha))| \le 6\cdot 2^{n-1}$ by Lemma~\ref{lem:estimateFib}.
\end{proof}

In the rest of the section, we prove the following result which is an equivalent version of Theorem~\ref{thm:AsymptoFib}. 

\begin{theorem}\label{thm:convergenceFib}
The function $\Psi$ defined in Proposition~\ref{pro:existence_limit} is continuous on $[0,1)$ such that $\Psi(0)=1$, $\lim_{\alpha\to 1^-}\Psi(\alpha)=1$ and the sequence $(A_F(N))_{N\ge 0}$ which is the summatory function of the sequence $(\mathsf{s}_F(n))_{n\ge 0}$ satisfies, for $N \geq 3$,
$$
	A_F(N)=c\, \beta^{\log_F N} \Psi(\relpos_F(N))+o(\beta^{\lfloor \log_F N \rfloor})
$$
where $\beta$ is the dominant root of $X^3-2X^2-X+1$.
\end{theorem}

A representation of $\Psi$ is given in Figure~\ref{fig:FF}. 
It has been obtained by estimating $A_F(N)/(c\, \beta^{\log_F N})$ for $N$ between $2584$ and $4180$. 
\begin{figure}[h!tb]
    \centering
    \scalebox{.7}{\includegraphics{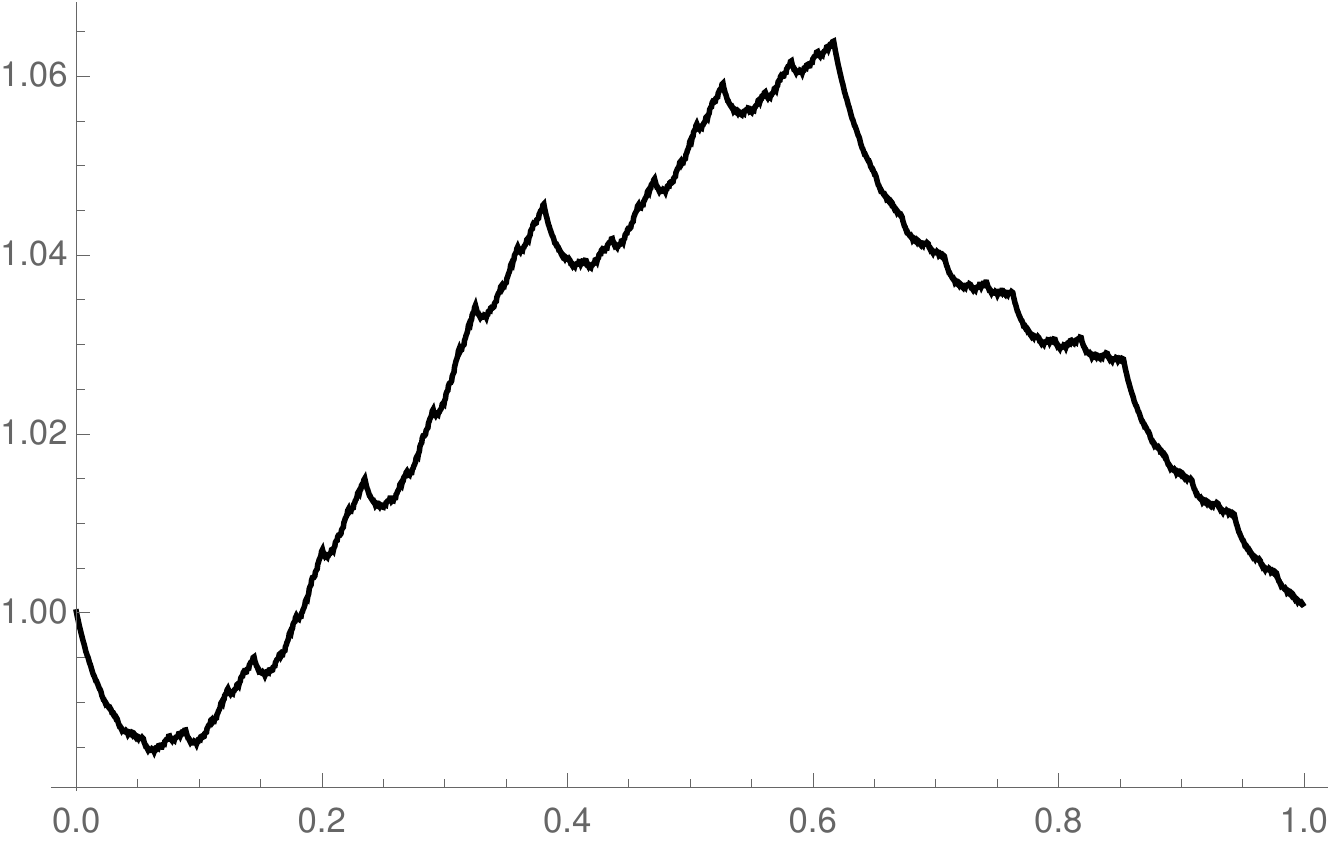}}
    \caption{The graph of $\Psi$.}
    \label{fig:FF}
\end{figure}

\begin{proof} This proof is divided into four parts: the error term for the sequence $(A_F(N))_{N\ge 0}$, the fact that $\Psi(0)=1$, the limit $\lim_{\alpha\to 1^-}\Psi(\alpha)=1$ and the continuity of the function $\Psi$. 

$\bullet$ We first focus on the error term. 
Let $\rep_F(N)=10r_1\cdots r_k$ with $k \geq 1$ and $r_1\cdots r_k \in\{1,\varepsilon\}\{0,01\}^*$. 
Observe that $k$ depends on $N$ since $k+2=|\rep_F(N)|$. 
By definition, we have
$$\relpos_F(N)=\sum_{i=1}^k \frac{r_i}{\varphi^i}.$$
On the one hand, Lemma~\ref{lem:equiv_rationnel} gives 
$$
c\, \beta^{\log_F N} \Psi(\relpos_F(N))
= c\, \sum_{i=0}^{k-1} b_{i}(N)\, \beta^{k+1-i} 
+ c\, \beta \, b_{k}(\relpos_F(N)) 
+ c \, b_{k+1}(\relpos_F(N)).
$$
On the other hand, we know that
$$A_F(N)=\sum_{i=0}^{k+1} b_{i}(N) B(k+1-i).$$
Thus, the error term is obtained by
\begin{eqnarray*}
R(N)&:=&A_F(N)-c\, \beta^{\log_F N} \Psi(\relpos_F(N))\\
    &=&\sum_{i=0}^{k-1} b_{i}(N) \left( B(k+1-i)-c\, \beta^{k+1-i}\right) \\
&&+b_k(N) B(1) - c\, \beta \, b_k(\relpos_F(N)) +
b_{k+1}(N) - c \, b_{k+1}(\relpos_F(N)).
\end{eqnarray*}
Let us divide the latter expression by $\beta^{k+1}$. 
Using \eqref{eq:recB}, we get
\begin{eqnarray*}
\frac{R(N)}{\beta^{k+1}}&=&\sum_{i=0}^{k-1} \frac{b_{i}(N)}{\beta^i}\ \frac{c_2\beta_2^{k+1-i}+c_3\beta_3^{k+1-i}}{\beta^{k+1-i}} 
+\frac{b_k(N) B(1) - c\, \beta\, b_k(\relpos_F(N))}{\beta^{k+1}} \\
&&+\frac{b_{k+1}(N) - c \, b_{k+1}(\relpos_F(N))}{\beta^{k+1}}.
\end{eqnarray*}
Firstly, we have
$$\frac{|c_2\beta_2^{k+1-i}+c_3\beta_3^{k+1-i}|}{\beta^{k+1-i}}\le 2 |c_2|\, \left(\frac{|\beta_2|}{\beta}\right)^{k+1-i} $$
and, from Lemma~\ref{lem:estimateFib}, 
$$\frac{|b_{i}(N)|}{\beta^i}\le 6 \left(\frac{2}{\beta}\right)^i.$$
Secondly, we have
$$\frac{|b_k(N) B(1) - c\, \beta\, b_k(\relpos_F(N))|}{\beta^{k+1}}
\le \frac{9 \cdot 2^{k+1} + c \, \beta\, 3 \cdot 2^{k+1}}{\beta^{k+1}} < 3\, (3+c\beta ) \left(\frac{2}{\beta}\right)^{k+1}$$
and
$$
\frac{|b_{k+1}(N) - c \, b_{k+1}(\relpos_F(N))|}{\beta^{k+1}} \le \frac{3\cdot 2^{k+2} +  c \, 3\cdot 2^{k+2}}{\beta^{k+1}} < 3\,(1+c)\, \beta \, \left(\frac{2}{\beta}\right)^{k+2}   
$$
again using Lemma~\ref{lem:estimateFib}. 
Hence
$$\frac{|R(N)|}{\beta^{k+1}}\le 
\frac{12|c_2|\, |\beta_2|^{2}}{\beta^{2}} \sum_{i=0}^{k-1} \left(\frac{2}{\beta}\right)^i\ \left(\frac{|\beta_2|}{\beta}\right)^{k-1-i} 
+ 3\,(3+c\beta ) \left(\frac{2}{\beta}\right)^{k+1}
+ 3\, (1+c)\, \beta \, \left(\frac{2}{\beta}\right)^{k+2}. 
$$
Since $\sum_{i=0}^{k-1} a^i b^{k-1-i}=(a^{k}-b^{k})/(a-b)$, we deduce that
$$
\frac{|R(N)|}{\beta^{k+1}}\le \frac{12|c_2|\, |\beta_2|^2}{\beta (2-|\beta_2|)} \,  \left(\left(\frac{2}{\beta}\right)^k -\ \left(\frac{|\beta_2|}{\beta}\right)^{k}\right)
+3\, (3+c\beta ) \left(\frac{2}{\beta}\right)^{k+1}
+ 3\, (1+c)\, \beta \, \left(\frac{2}{\beta}\right)^{k+2}, 
$$
which tends to zero when $k$ tends to infinity since $\beta > \beta_2$ and $\beta > 2$. 
This implies that $R(N)=o(\beta^{k+2})$.

$\bullet$ We show that $\Psi(0)=1$. 
By definition, $w_n(0)$ is the prefix of length $n$ of the infinite word $10\rep_\varphi(0)=10^\omega$ and $e_n(0)$ is thus equal to $F(n-1)$. 
By definition of $\Psi$ and using \eqref{eq:BdecompFib}, we have
$$
\Psi(0) = \lim_{n\to +\infty} \frac{B(n-1) - B(1) + (n+4) \, B(0)}{c\, \beta^{n-1}} =1
$$
since $\relpos_F(e_n(0))=0$.

$\bullet$ 
To show that 
$$
\lim_{\alpha\to 1^-} \Psi(\alpha)= 1,
$$
we make use of the uniform convergence and consider 
\begin{eqnarray*} 
\lim_{\alpha\to 1^-} \lim_{n\to +\infty} 
\frac{1}{\beta^{\relpos_F(e_n(\alpha))}}
\frac{A_F(e_n(\alpha))}{c\, \beta^{n-1}} 
&=& 
\lim_{\alpha\to 1^-} \frac{1}{\beta^{\alpha}} 
\lim_{n\to +\infty} 
\frac{A_F(e_n(\alpha))}{c\, \beta^{n-1}} 
\\
&=&
\frac{1}{\beta}
\lim_{\alpha\to 1^-} \lim_{n\to +\infty} 
\frac{A_F(e_n(\alpha))}{c\, \beta^{n-1}}
\\
&=&
\frac{1}{\beta}
\lim_{n\to +\infty} \lim_{\alpha\to 1^-} 
\frac{A_F(e_n(\alpha))}{c\, \beta^{n-1}}. 
\end{eqnarray*}For any fixed integer $n \geq 3$, we can chose $\alpha$ close enough to 1 such that 
\[
	\rep_\varphi(\alpha) \in (10)^n \{0,1\}^\omega.
\]
Using~\eqref{eq:Bded de 10101010}, we have 
\begin{eqnarray*}
	A_F(e_n(\alpha)) 
	&=& A_F(F(n)-1) \\ 
	&=& 
		\sum_{i=0}^{n-3} g_i B(n-1-i) + b_{n-2}(e_n(\alpha)) B(1) + b_{n-1}(e_n(\alpha)) B(0). 
\end{eqnarray*}
Due to Lemma~\ref{lem:Bdecomp_2ex} and Lemma~\ref{lem:estimateFib}, we have $|g_k| \leq 2 (\sqrt{2})^k$ for all $k$ and both $b_{n-1}(e_n(\alpha))$ and $b_{n-2}(e_n(\alpha))$ are smaller than $3\cdot 2^{n}$. 
Hence we have 
\[
	\lim_{n\to +\infty} \frac{b_{n-2}(e_n(\alpha))B(1)+b_{n-1}(e_n(\alpha))B(0)}{c\beta^{n-1}} = 0.
\]
Our aim is thus to show that 
\[
\lim_{n\to +\infty} \sum_{i=0}^{n-3} \frac{g_i}{\beta^i} \frac{B(n-1-i)}{c \beta^{n-1-i}} = \beta.
\]
By Lemma~\ref{lem:Bdecomp_2ex}, we have
\begin{eqnarray*}
\left|
\sum_{i=0}^{n-3} \frac{g_i}{\beta^i} \frac{B(n-1-i)}{c \beta^{n-1-i}} - \beta
\right|
&\leq&
\left|
\sum_{i=0}^{n-3} \frac{g_i}{\beta^i} \left(\frac{B(n-1-i)}{c \beta^{n-1-i}}-1\right)
\right|
+ 
\left|
\sum_{i=n-2}^{+\infty} \frac{g_i}{\beta^i}
\right|.
\end{eqnarray*}
Again by Lemma~\ref{lem:Bdecomp_2ex}, 
$\left|
\sum_{i=n-2}^{+\infty} \frac{g_i}{\beta^i}
\right|$
goes to 0 as $n$ goes to infinity.
Using~\eqref{eq:recB}, we have 
\[
\left|\frac{B(n-1-i)}{c \beta^{n-1-i}}-1\right|
\leq 
\frac{2|c_2||\beta_2|^{n-1-i}}{c\beta^{n-1-i}}
\]
and thus
\begin{eqnarray}\label{eq:thelastone}
\left|
\sum_{i=0}^{n-3} \frac{g_i}{\beta^i} \left( \frac{B(n-1-i)}{c \beta^{n-1-i}} - 1 \right)
\right|
& \leq &  
\frac{4|c_2||\beta_2|^{n-1}}{c \beta^{n-1}} 
\sum_{i=0}^{n-3} \left(\frac{\sqrt{2}}{|\beta_2|}\right)^i
\\
& \leq & 
\frac{4|c_2||\beta_2|^2}{c (\sqrt{2}-|\beta_2|)} \frac{(\sqrt{2})^{n-2}-|\beta_2|^{n-2}}{\beta^{n-1}},\nonumber
\end{eqnarray}
which also tends to $0$ as $n$ tends to infinity.
This shows that $\lim_{\alpha\to 1^-}\Psi(\alpha) = 1$.

$\bullet$ To finish the proof, let us show that $\Psi$ is continuous.
Let $\alpha \in [0,1)$ and let us write $\rep_\varphi(\alpha) = (d_n)_{n \geq 1}$. 
We make use of the uniform convergence of the sequence $(\psi_n)_{n \in \mathbb{N}}$ and consider
\begin{eqnarray*}
	\lim_{\gamma \to \alpha} 
	|\Psi(\gamma) - \Psi(\alpha)|
	&=& \lim_{\gamma \to \alpha} \lim_{n \to +\infty} 
	|\psi_n(\gamma)-\psi_n(\alpha)|	\\
	&=& \lim_{n \to +\infty} \lim_{\gamma \to \alpha}
	|\psi_n(\gamma)-\psi_n(\alpha)|
\end{eqnarray*}
First assume that $\alpha$ is not of the form $\sum_{i=1}^k r_i/\varphi^i$ , i.e., $(d_n)_{n \geq 1}$ does not belong to $\{0,1\}^*10^\omega$.
For any fixed integer $n$, we can chose $\gamma$ close enough to $\alpha$ such that $\rep_\varphi(\gamma) \in d_1 d_2 \cdots d_n \{0,1\}^\omega$.
Therefore, we have $e_n(\gamma) = e_n(\alpha)$, hence $\psi_n(\gamma) = \psi_n(\alpha)$.

Now assume that $\rep_\varphi(\alpha) = d_1d_2 \cdots d_k 0^\omega$ with $d_k = 1$.
For any fixed integer $n > k+1$, we can chose $\gamma$ close enough to $\alpha$ such that
\[
\begin{array}{ll}
\rep_\varphi(\gamma)  \in  d_1d_2 \cdots d_k 0^n \{0,1\}^\omega,
& \text{if }  \gamma \geq \alpha; \\
\rep_\varphi(\gamma)  \in  d_1d_2 \cdots d_{k-1} (01)^n \{0,1\}^\omega,
& \text{if } \gamma < \alpha.

\end{array}
\]
If $\gamma \geq \alpha$, we get $\psi_n(\gamma) = \psi_n(\alpha)$ as in the first case.
If $\gamma < \alpha$, we get 
\begin{eqnarray*}
	e_n(\alpha) &=& \val_F(10d_1d_2\cdots d_k 0^{n-k-2}), \\
	e_n(\gamma) &=& \val_F(10d_1d_2\cdots d_{k-1} (01)^{\frac{n-k-1}{2}}),
\end{eqnarray*}
where fractional powers of words are classically defined by, for $w =w_1 w_2 \cdots w_{|w|}$, $w^{p/|w|} = w^\ell w_1 w_2\cdots w_q$ if $p/|w| = \ell+q/|w|$ with $0 \leq q < |w|$. 

In this case we get $e_n(\alpha) = e_n(\gamma)+1$ and
\begin{eqnarray*}
	|\psi_n(\alpha) - \psi_n(\gamma)| 
	&\leq &
	\left| \frac{A_F(e_n(\alpha))}{c \beta^{n-1}} 
	\left( \frac{1}{\beta^{\relpos_F(e_n(\alpha))}} - \frac{1}{\beta^{\relpos_F(e_n(\gamma))}}
	\right) 
	\right| \\
	&+&
	\left|
	\frac{1}{c\beta^{\log_F(e_n(\gamma))}} 
	\left( A_F(e_n(\alpha)) - A_F(e_n(\gamma)) \right)
	\right|.
\end{eqnarray*}
Using~\eqref{eq:diff relp e_n alpha}, we have $|\relpos_F(e_n(\alpha)) - \relpos_F(e_n(\gamma))| < 3/\varphi^{n-2}$.
As $A_F(e_n(\alpha)) / (c \beta^{n-1})$ converges to some real number $x$ when $n$ goes to infinity, the first term tends to zero when $n$ increases.
The second term also tends to zero as, by Corollary~\ref{cor:majsF}, $A_F(e_n(\alpha)) - A_F(e_n(\gamma)) = s_F(e_n(\alpha)) \leq 2^n$. 
This shows that $\Phi$ is continuous.
\end{proof}

\section{Possible extensions to other numeration systems}

One can wonder whether the method presented in this paper can be applied to classical digital sequences. 
Consider the example of the sum-of-digits function $s_2$ for base-$2$ expansions of integers mentioned in the introduction. 
Its summatory function $(A_2(n))_{n \geq 0}$ satisfies 
\[
	A_2(2^\ell+r) = U(\ell) + A_2(r) + r \cdot U(1)
\]
for all $\ell \geq 0$, where $U(0) = 0$, $U(1) = 1$, and the sequence  $(U(n))_{n \geq 0}$ satisfies the linear recurrence relation 
\[
	U(n+2) = 4\, U(n+1) - 4\, U(n) \quad \forall n\ge 0.
\]
Numerical experiments suggest that our method gives the same result as~\eqref{eq:delange}.

Other examples can be considered with sequences defined analogously to $(\mathsf{s}(n))_{n \geq 0}$ and $(\mathsf{s}_F(n))_{n \geq 0}$, i.e., to consider sequences associated with binomial coefficients of representations of integers in some numeration systems.
The main problem is that we do not have a statement similar to Proposition~\ref{pro:rec} or Proposition~\ref{pro:recF}. Nevertheless, we proceeded to some computer experiments. 
For an integer base $k\ge 3$, let $A_k(n)$ denote the analogue of $A(n)$ when we consider words and subwords in the language $\{1,\ldots,k-1\}\{0,\ldots,k-1\}^*\cup\{\varepsilon\}$. 
In that case, we conjecture that 
$$A_k(kn)=(2k-1)A_k(n), \quad \forall n\ge 1$$
and there exists a continuous and periodic function $\mathcal{H}_k$ of period 1 such that 
$$
A_k(N)=(2k-1)^{\log_k(N)}\mathcal{H}_k(\log_k(N)).
$$
The graphs of $\mathcal{H}_3,\mathcal{H}_4,\mathcal{H}_5,\mathcal{H}_7$ have been depicted in Figure~\ref{fig:conj} on the interval $[0,1)$. 
Such a result is in the line of~\eqref{eq:delange}.
\begin{figure}[h!t]
    \centering
    {\scalebox{.3}{\includegraphics{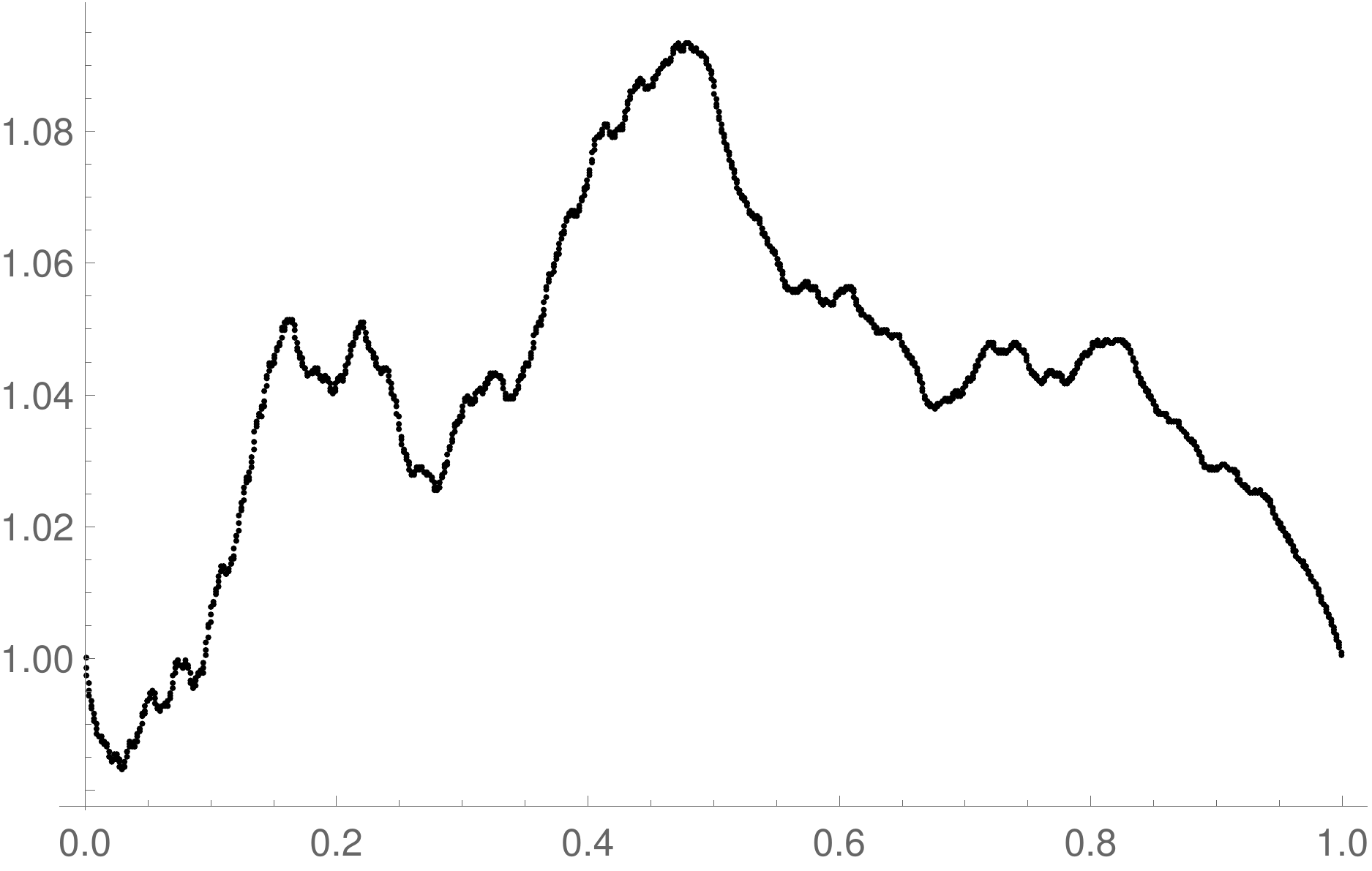}}}\quad {\scalebox{.3}{\includegraphics{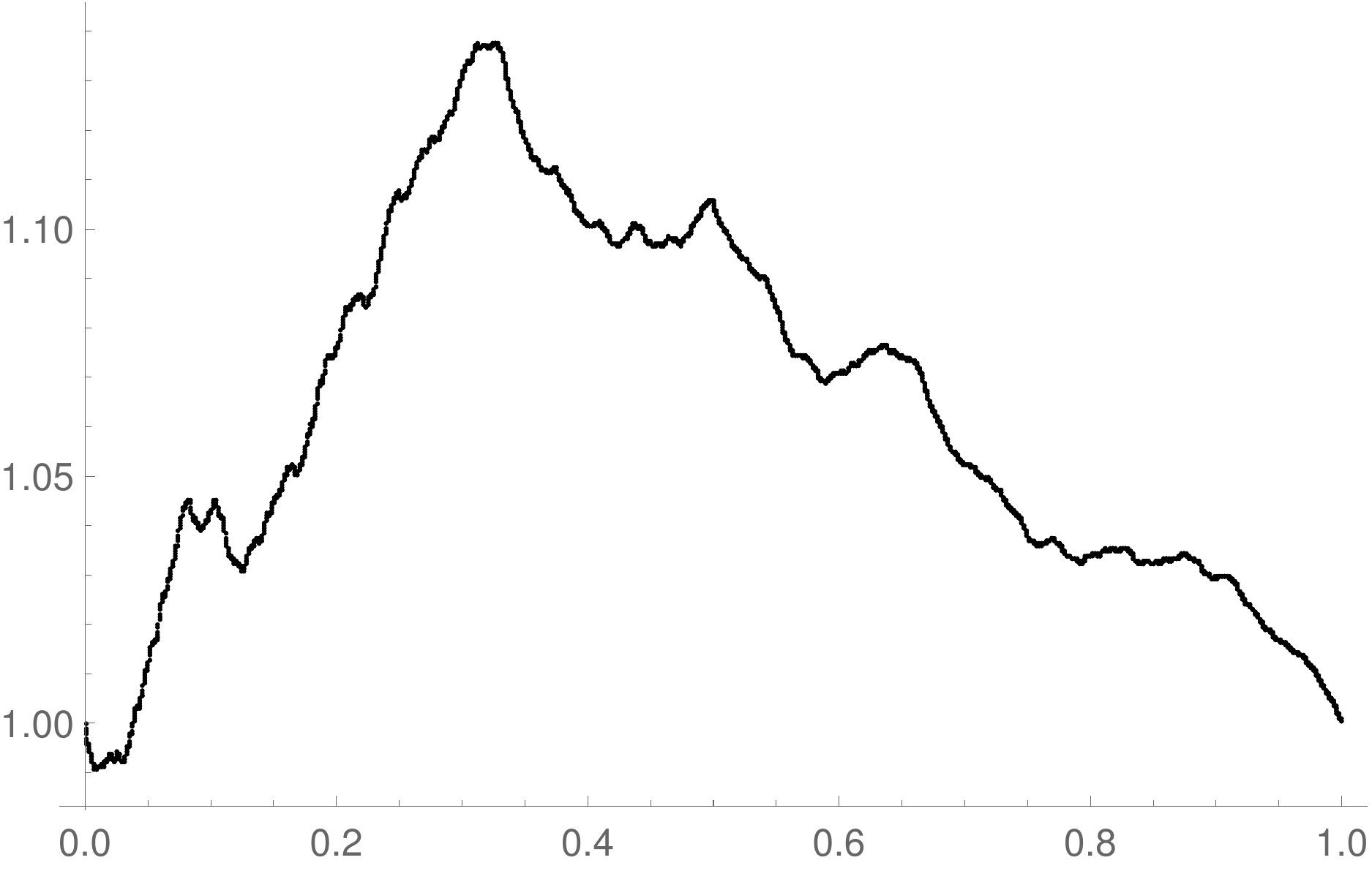}}}\\
{\scalebox{.3}{\includegraphics{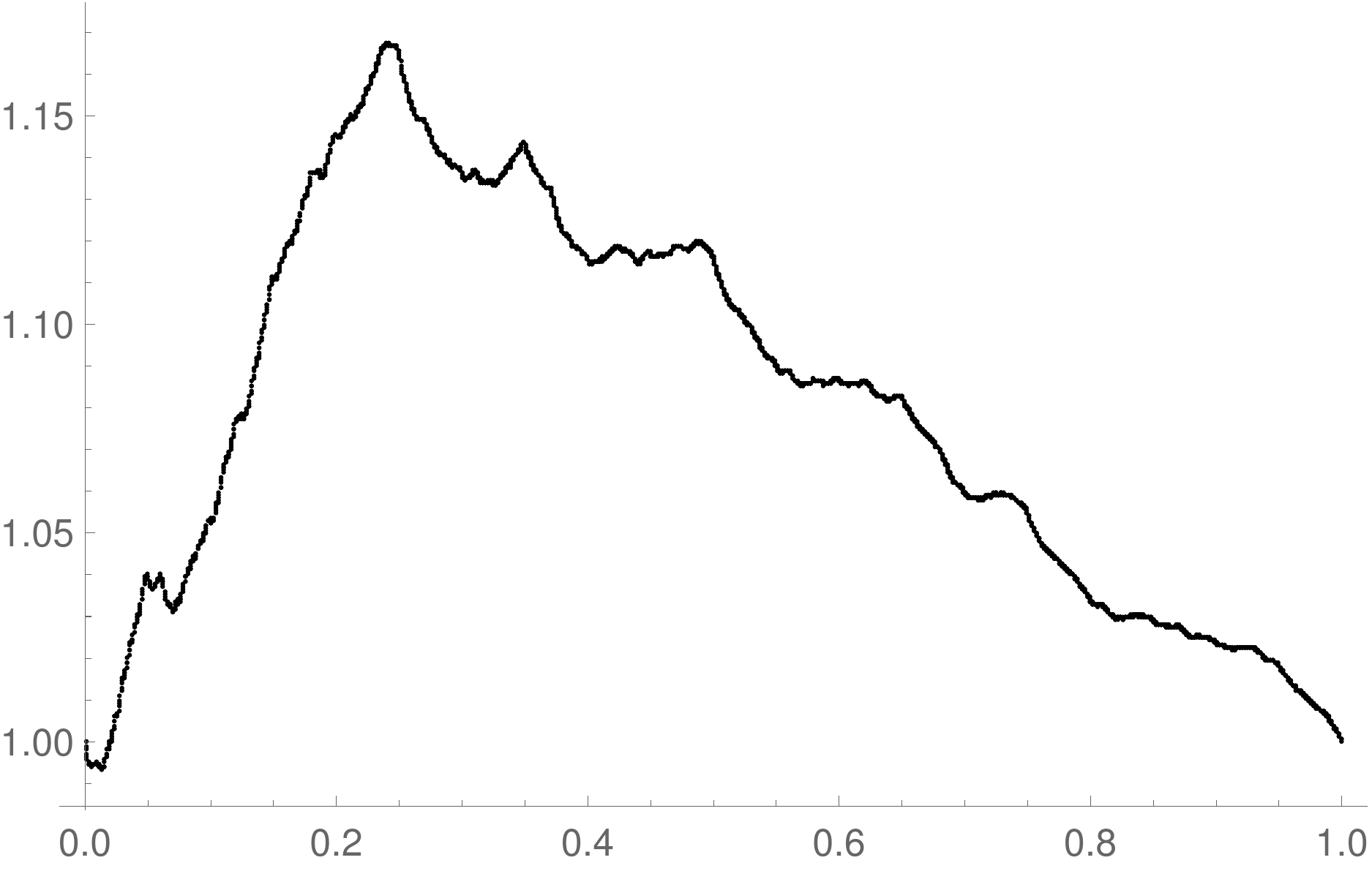}}}\quad {\scalebox{.3}{\includegraphics{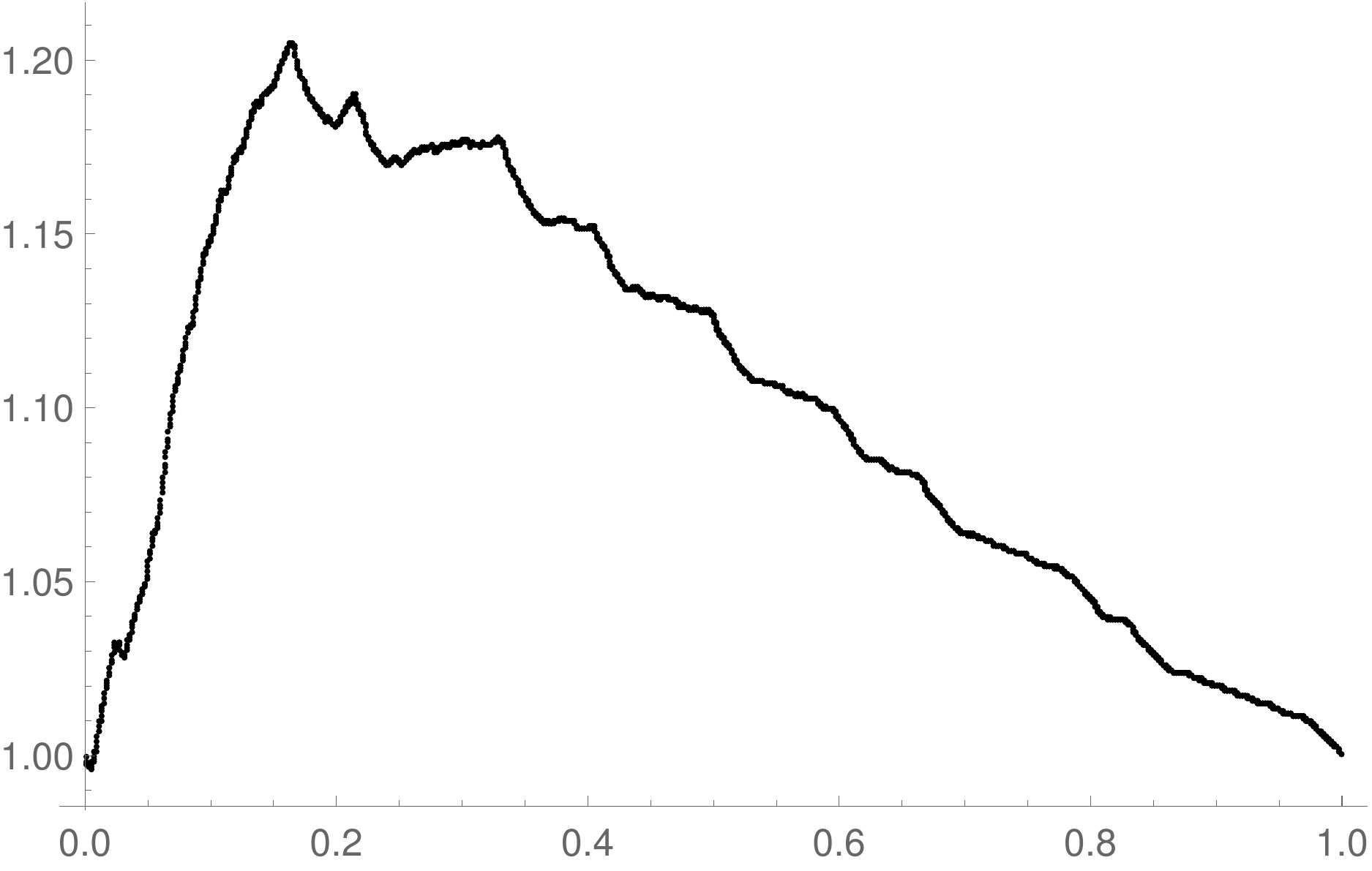}}}
    \caption{The conjectured functions $\mathcal{H}_3,\mathcal{H}_4,\mathcal{H}_5,\mathcal{H}_7$.}
    \label{fig:conj}
\end{figure}

If we leave the $k$-regular setting and try to replace the Fibonacci sequence with another linear recurrent sequence, the situation seems to be more intricate. 
For the Tribonacci numeration system $T=(T(n))_{n\ge 0}$ built on the language of words over $\{0,1\}$ avoiding three consecutive ones, we conjecture that a result similar to Theorem~\ref{thm:AsymptoFib} should hold for the corresponding summatory function $A_T$. 
Computing the first values of $A_T(T(n))$, the sequence $(B(n))_{n\ge 0}$ should be replaced with the sequence $(V(n))_{n\ge 0}$ satisfying 
$$V(n+5)=3V(n+4)-V(n+3)+V(n+2)-2V(n+1)+2V(n) \quad \forall n\ge 0$$
and with initial conditions $1,3,9,23,63$. 
The dominant root $\beta_T$ of the characteristic polynomial of the recurrence is close to $2.703$.
\begin{figure}[h!t]
    \centering
    {\scalebox{.3}{\includegraphics{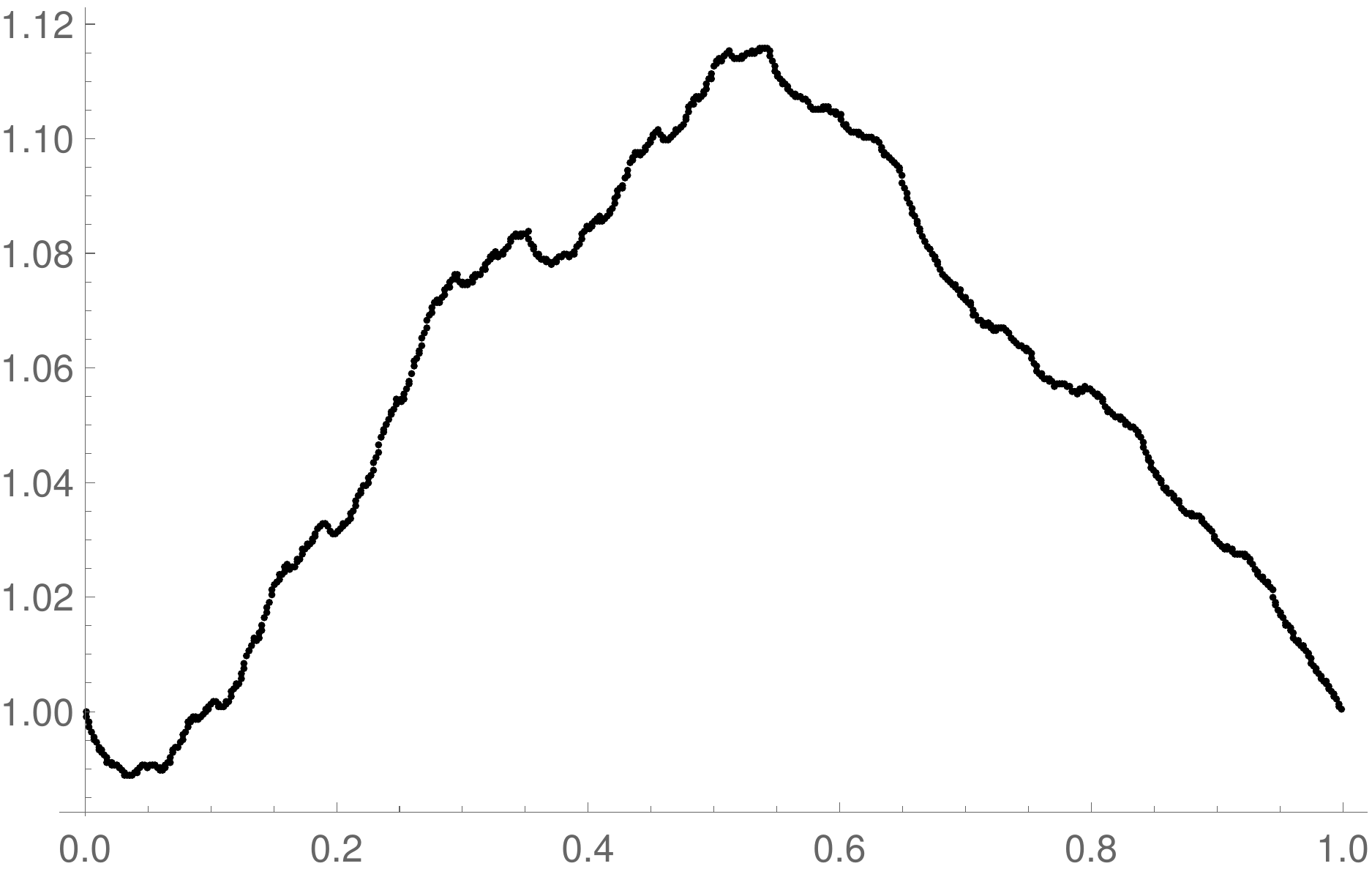}}}\quad {\scalebox{.3}{\includegraphics{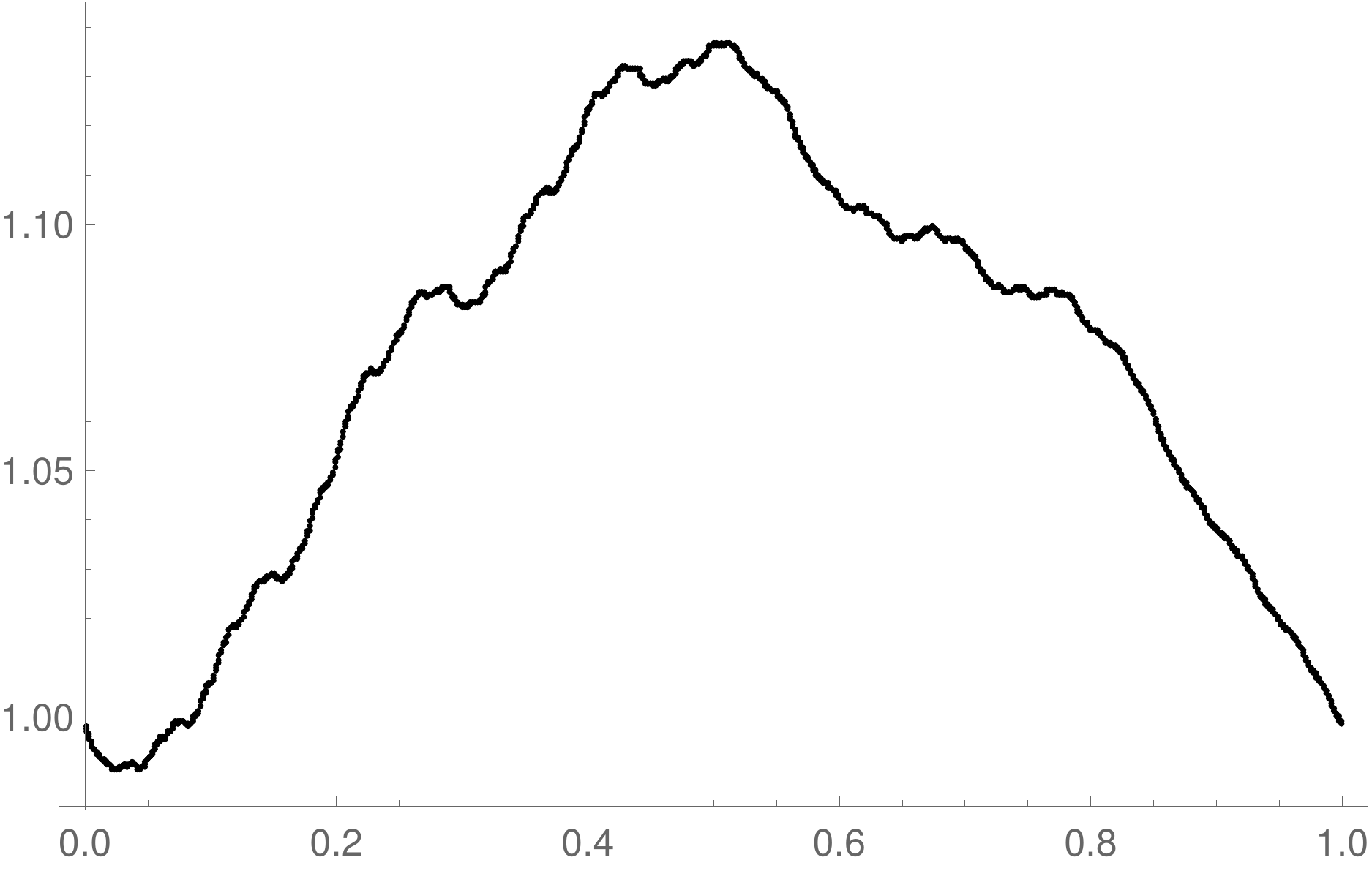}}}
    \caption{The conjectured functions $\mathcal{G}_T$ and $\mathcal{G}_Q$.}
    \label{fig:conj2}
\end{figure}
There should exist a continuous and periodic function $\mathcal{G}_T$ of period 1 whose graph is depicted in Figure~\ref{fig:conj2} such that the corresponding summatory function has a main term in $c_T\, \beta_T^{\log_T(N)} \mathcal{G}_T(\log_T(N))$ where the definition of $\log_T$ is straightforward. 
We are also able to handle the same computations with the Quadribonacci numeration system where the factor $1^4$ is avoided.
In that case, the analogue of the sequence $(B(n))_{n\ge 0}$ should be a linear recurrent sequence of order $6$ whose characteristic polynomial is $X^7-4 X^6+4 X^5-2 X^4-X^3+3 X^2-6 X+2$. 
Again, we conjecture a similar behavior with a function $\mathcal{G}_Q$ depicted in Figure~\ref{fig:conj2}. 
Probably, the same type of result can be expected for Pisot numeration systems (i.e., linear recurrences whose characteristic polynomial is the minimal polynomial of a Pisot number).

\subsection*{Acknowledgements}
We thank the anonymous referee for the feedback and suggestions to improve the presentation of the paper.

\end{document}